\documentclass[12pt]{article}
\usepackage{rotating,multirow,amsmath,color,amssymb,fullpage,MnSymbol}
\usepackage{natbib}
\usepackage{graphicx}
\usepackage{subfigure}
\usepackage{float}
\usepackage{tikz}
\newtheorem{conj}{Conjecture}
\usepackage{subfigure}
\usepackage{float}
\usepackage{graphicx}
\usepackage{blindtext}

  \usepackage{mathtools}

\bibpunct[, ]{(}{)}{,}{a}{}{,}%

\newtheorem{definition}{Definition}
\newtheorem{theorem}{Theorem}
\newtheorem{lemma}{Lemma}

\newtheorem{corollary}{Corollary}
\newtheorem{example}{Example}

\newenvironment{proof}{\paragraph{Proof:}}{\hfill$\Box$}

\title{Continuous Patrolling Games}

\author{Steve Alpern\thanks{Warwick Business School, University of Warwick, Coventry CV4 7AL, UK, steve.alpern@wbs.ac.uk}  \and Thuy Bui\thanks{Rutgers Business School, Newark, NJ 07102, USA, tb680@business.rutgers.edu} \and Thomas Lidbetter\thanks{Rutgers Business School, Newark, NJ 07102, USA, tlidbetter@business.rutgers.edu} \thanks{Department of Engineering Systems and Environment, University of Virginia, Charlottesville VA 22904} \and Katerina Papadaki\thanks{Department of Mathematics, London School of Economics, London WC2A 2AE, UK, k.p.papadaki@lse.ac.uk}}

\begin{document}
	
	\maketitle

	\begin{abstract}
\noindent We study a patrolling game played on a network $Q$, considered as
a metric space. The Attacker chooses a point of $Q$ (not necessarily a node)
to attack during a chosen time interval of fixed duration. The Patroller
chooses a unit speed path on $Q$ and intercepts the attack (and wins) if she
visits the attacked point during the attack time interval. This zero-sum
game models the problem of protecting roads or pipelines from an adversarial
attack. The payoff to the maximizing Patroller is the probability that the
attack is intercepted. Our results include the following: (i) a solution to
the game for any network $Q$, as long as the time required to carry out the
attack is sufficiently short, (ii) a solution to the game for all tree
networks that satisfy a certain condition on their extremities, and (iii) a
solution to the game for any attack duration for stars with one long arc and
the remaining arcs equal in length. We present a conjecture on the solution
of the game for arbitrary trees and establish it in certain cases.
	\end{abstract}

	\newpage

\section{Introduction} \label{sec:intro}

Patrolling games were introduced at the end of \cite{AMP} to model the operational
problem of how to optimally schedule patrols to intercept a terrorist
attack, theft or infiltration. That paper, contrasting with earlier
adversarial patrolling (Stackelberg) versions, modeled the problem as a
zero-sum game between an Attacker and a Patroller, who wish to respectively
maximize and minimize the probability of a successful attack. The domain on
which the game was played out was taken to be a graph, with attacks
restricted to the nodes and taking a given integer number of periods. A patrol is a
walk on the graph, and intercepts the attack if it visits the attacked node
during the attack period. This could model a guard in an art museum who enters a room while a thief is in the midst of removing a valuable
painting from the wall. That paper was able to make some key observations
about their game, giving bounds on the value, but was unable to find the
value precisely or give optimal strategies except in some very limited
cases. \cite{PALM16} solved the game for line graphs, but the solution was
very complicated even for this apparently simple graph. In the
Conclusion section of the original paper \cite{AMP}, an extension of the
problem to continuous space and time was suggested. The purpose of this paper is to carry out this suggestion.

We allow attacks that have a prescribed duration~$\alpha$ to occur at any point of a continuous network $Q$. A unit speed patrol on $Q$ is said to intercept
the attack (and win for the Patroller) if it arrives at the attacked point at some
time during the attack. The value of the game is the probability of interception,
with best play on both sides. We find that optimal play for the Attacker typically involves mixing pure
attacks that take place at different times.

After this type of continuous game was first proposed in 2011,
it has been solved for some special networks. The circle network (or any Eulerian
network) is easy to solve: a periodic traversal of the Eulerian tour,
starting at a random point, is optimal for the Patroller; attacking starting
at a fixed time at a uniformly random location is optimal for the Attacker
(see \cite{ALMP} and \cite{Garrec}). The line segment network was solved in
\cite{ALMP}. In \cite{Garrec} a
solution for some values of $\alpha$ is given for the network with two nodes connected by three unit
length arcs, and a complete formulation of the general game is given,
including a proof of the existence of the value. The present paper extends to some extent all three of these prior results to general classes of networks: Eulerian networks to networks without leaf arcs; the line segment network to
trees; the three-arc network to networks with large girth - for small attack
times. The Area Editor has observed that ``in real-life the attacker has no incentive to hang out at the attack site - he would disappear as fast as he can following the attack. Therefore, small $\alpha$ is reasonable for many real-life situations.''

Our main results and chapter organization are as follows. Section~\ref%
{sec:definitions} presents several (mixed) strategies for the players that
can be used or adapted to obtain solutions of the game for various classes
of networks in later sections. We note that Eulerian networks have no
leaves, and Section~\ref{sec:noleaves} generalizes the solution of the
former to networks without leaves. In particular, as long as the attack time
is sufficiently short, we show that the attack strategy that chooses a point
uniformly at random is still optimal; an optimal strategy for the Patroller
is to follow a double cover tour of the network which never traverses
an arc consecutively in opposite directions (as described in Theorem~\ref%
{theorem:network_noleaves}). We also give a new algorithm for constructing
such a tour in Theorem~\ref{thm:opt-tour}. In Section~\ref{sec:arbitrary} we
allow the network to have leaves, and modify the optimal strategies of the
previous section to generate optimal strategies for arbitrary networks, as
long as the attack time is sufficiently short (see Theorem~\ref%
{thm:gen-girth}).

Section~\ref{sec:trees} considers trees and in particular those that satisfy
a condition we call the Leaf Condition. We give a precise definition of the condition, which requires some delicacy (Definition~\ref{def:leaf-cond}). In fact, any tree satisfies the Leaf Condition as long as the attack time is sufficiently short. Star networks
(trees with only leaf arcs) also satisfy the Leaf Condition for sufficiently
large attack times, and the only stars that do not satisfy the Leaf
Condition are those that have an arc that is longer than half the total
length of the network. In Theorem~\ref{theorem:trees} we solve the game for
all trees in the case that the Leaf Condition holds, giving a simple
expression for the value of the game in terms of the length of the network,
the attack time and another parameter. In Subsection~\ref{sec:trees2} Conjecture~\ref{conj} states that this expression is always equal to the value of the game
on trees. We establish the conjecture for some stars that do not
satisfy the Leaf Condition.

\section{Literature Review}

In addition to the papers discussed in the Introduction, which were the most
relevant to continuous patrolling, there is a more extensive literature on
adversarial patrolling. The problem of patrolling a perimeter has been
analyzed by \cite{Zoroa12} (where the attack location can move to adjacent
locations) and \cite{Lin19}, the latter in a continuous time context.
Extensions of \cite{ALMP} where the costs of successful attacks are time and
node dependent have been studied by \cite{Lin13} (for random attack times),
\cite{Lin14} (with imperfect detection) and \cite{Yolmeh19} (which includes
an application to an urban rail network).

Stackelberg approaches, with the Patroller as first mover, have been
pioneered in an artificial intelligence context by \cite{Basilico12} (which includes
an algorithm for large cases)  and \cite{Basilico17} (where the optimal
strategy in certain cases is for the Patroller to stay in place until the
sensor reveals an attack an unknown location).

More applied approaches to patrolling are of practical importance.
Applications to scheduling randomized security checks and canine patrols at
Los Angeles Airport have been developed and deployed in \cite{ARMOR}. The
United States Coast Guard also uses a game-theoretic system to schedule
patrols in the Port of Boston \citep{Coastguard}. Recently, a game theoretic
approach to schedule patrols to guard against poachers has been explored in
\cite{Poaching16} (where the novel algorithm PAWS was introduced) and \cite%
{Poaching19} (where the success of deploying PAWS in the field is
described). Patrolling to detect radiation and consequently nuclear threats was
modeled in the novel paper of  \cite{Hochbaum14}.

The possibility that the Attacker could know when the Patroller is nearby
(perhaps at the same node), raised in \cite{AMP}, has recently been studied
in \cite{AlpernKatsikas19}, \cite{AlpernKatsikas21} and \cite{Lin19} in different contexts. In the
former this knowledge helped the Attacker, in the latter, it did not.
Multiple patrollers have been considered in the robotics and computer
science literatures, where an important paper with a similar network
structure to ours is \cite{Czyzowicz17}.  A connection between patrols and
inspection games is made in \cite{Baston91} and between patrols and
hide-seek games in \cite{Garrec}. Restricting the Patroller to periodic
paths creates difficulties analyzed in \cite{ALP}.

\section{Formal Definitions for Network and Game}
\label{sec:definitions}

In this section we define the continuous patrolling game and present
definitions related to the connected network $Q$ on which it is played. For $Q$,
standard graph theoretic definitions must be modified for a network which is
considered as a metric space and a measure space, not simply a combinatorial
object.

To define $Q$, we begin with a graph $G$ with edges and
vertices, with the addition of a \textit{length} $\lambda \left(e\right) $
assigned to each edge $e$. We can then identify an edge $e$ with an open interval
of length $\lambda \left(e\right) $, endowed with Lebesgue measure and
Euclidean distance $d$, and consider $\lambda $ as a measure on $Q$, called
\textit{length}. The total length of $Q$ is denoted by $\mu =\lambda \left(
Q\right) $. The topology on these intervals gives a topology on their
union $Q$. A \textit{path} in $Q$ is a continuous function from a closed
interval to $Q$. We take the metric $d\left( x,y\right) $ on $Q$ as the
minimum length of a path between $x$ and $y$. A point $x$ of $Q$ is called a
\textit{regular} point if it has a neighborhood homeomorphic to an open
interval. The remaining non-regular points are called {\em nodes}. The {\em degree} of a point $y$ is defined as the number of connected components of a small neighborhood of $y$ after $y$ has been removed from it. Such a neighborhood is called a {\em punctured neighborhood} in the topology literature. A point of degree $2$ is always by definition regular, and hence not a node. We say that two nodes of $Q$ are {\em adjacent} if there is a path between them consisting only of regular points.  Such a path is called an {\em arc}.
A node of degree $1$ is called a \textit{leaf node}, and its incident arc is called a \textit{leaf arc}.
To ensure that every leaf arc has a single leaf node in its closure, we exclude
the line segment network from consideration. In any case the continuous
patrolling game has been solved for the line segment in \cite{ALMP}.

A \textit{circuit} in $Q$ is a closed path (that is, with the same startpoint and endpoint) consisting of distinct adjacent
arcs. A {\em tour} of $Q$ is a closed path visiting all points of $Q$, and a tour
of minimum length is called a \textit{Chinese Postman Tour (CPT)}. The
length of this path is denoted $\bar{\mu}$. It was shown by Edmonds and
Johnson (1973) that a CPT can be found in polynomial time, with respect to the
number of nodes. A closed path which is a circuit and a tour is called an
Eulerian tour. As is well known, a connected network has an Eulerian tour if
and only if it is Eulerian, defined as having nodes all of even degree. If
we double every arc of a network $Q$, the resulting network is Eulerian with
length $2\mu $, so $Q$ has a tour of length $2\mu $ and hence $\bar{\mu}\leq
2\mu $.

The continuous patrolling game is played on $Q$ as follows. The Attacker
chooses a point $x$ in $Q$ to attack, and a closed time interval $J$ of
given length $\alpha $ during which to attack it. Since $\alpha $ is fixed,
the \textit{attack interval }$J=\left[ \tau ,\tau +\alpha \right] $ is
determined by its starting time $\tau.$ The game and its value are
determined by the pair $\left( Q,\alpha \right) $. The Patroller chooses a
path $S\left( t\right) $, where $t\geq 0$, which we call a \textit{patrol},
satisfying
\begin{align}
d\left( S\left( t\right) ,S\left( t^{\prime }\right) \right) \leq \left\vert
t-t^{\prime }\right\vert ,\text{ for all }t,t^{\prime }\geq 0.
\label{Lipshitz}
\end{align}
For simplicity, we shall call a path satisfying the $1-$Lipshitz condition (%
\ref{Lipshitz}) a \textit{unit speed path}. We don't specify an upper bound
on the starting time of the attack, but in every case we have studied there
is an optimal mixed attack strategy in which all its (pure strategy) attacks
are over by time $4\mu $. A patrol is said to \textit{intercept} an attack
if it visits the attacked point while it is being attacked. The game is very
simply defined: the maximizing Patroller wins (payoff $P=1$) if her patrol
intercepts the attack. Otherwise, the Attacker wins (payoff $P=0$ to the
Patroller). The payoffs to the Attacker are reversed, so the game has
constant sum $1$. In other words, if the patrol is $S$ and the
attack is at point $x$ during the interval $J=\left[ \tau ,\tau +\alpha %
\right] $, then the payoff $P$ to the maximizing Patroller is given by%
\[
P\left( S,\left( x,J\right) \right) =\left\{
\begin{array}{cc}
1 & \text{if }x\in S\left( J\right) , \\
0 & \text{otherwise.}%
\end{array}%
\right.
\]%
For mixed strategies, the expected payoff can be interpreted as the probability that the attack is intercepted. The value of the game, denoted $V$, is the interception probability, with best play on both sides.

\cite{Garrec} used the fact that $P$ is lower semicontinuous to establish
the existence of a value~$V$ for this infinite game. We note that if $\alpha
=0$ then the Attacker can win almost surely by attacking uniformly on $Q$
(according to $\lambda $) at a fixed time; if $\alpha \geq \bar{\mu}$, the
Patroller can ensure a win by adopting a Chinese Postman Tour, starting
anywhere at time $0$ and repeating the tour with period $\bar{\mu}$. So to
avoid the trivial cases where one of the player can always win, we assume $%
0<\alpha < \bar{\mu} $.

We follow \cite{Garrec} in not imposing a finite time horizon.  However, if we require that the attack ends by some time $T>\alpha$, this is only a restriction on the Attacker's strategy set. Hence, all Patroller estimates (lower bounds on the value) would remain valid. Attacker estimates (upper bounds on the value) also remain valid for sufficiently large $T$ because all the optimal Attacker strategies presented in this paper end by a stated finite time. For example, the uniform attack strategy, discussed in the next subsection, ends by time $M+\alpha$, where $M$ can be chosen arbitrarily.

Throughout the paper the complement $Q-Y$ of a set $Y$ is denoted by $Y^{c}$.

\subsection{The Uniform and the Independent Attack Strategies}

Some networks, as we shall see in later sections, require Attacker
strategies specifically suited to their structure, such as attacks on leaf
nodes when the network is a tree. But there are also some general strategies
that are available on any network. Here we define two of these and present
the general bounds on the value that they give.

\begin{definition} [\textbf{Uniform attack strategy}]
\label{def:uniform}
A \textbf{uniform attack strategy} is a mixture of pure attacks that have a common attack time interval  $J=[M,M+\alpha]$, where $M$ can be chosen arbitrarily (for example $M=0$). The attacked point is chosen uniformly at random. That is, the probability that the attacked point lies in a set $Y$ is given by $\lambda \left( Y\right) /\mu$.
\end{definition}

We restate a lemma from \cite{ALMP} for completeness (the proof is in the Appendix).
\begin{lemma} \label{lem:uniform}
Against any patrol $S$, a uniform attack strategy is intercepted with probability not more than $\alpha /\mu $. Consequently $V\leq \alpha
/\mu $ for any network.
\end{lemma}

We now define independence for sets and strategies.
\begin{definition}[\textbf{Independent set}]
A subset $I$ of $Q$ is called  \textbf{independent}  if the distance between
any two of its points is at least $\alpha$.
For any subset $Y$ of $Q$, the set $W \equiv W\left( Y\right) $ is the subset of $Q$ consisting of all points at
distance at most $\alpha /2$ from $Y$.
\end{definition}

\begin{definition}[\textbf{Independent attack strategy}]
Given an independent set $I$ of cardinality $l$ and the set $W \equiv W(I)$, the \textbf{independent attack strategy} is as follows for  $p=\frac{l \alpha}{\lambda \left( W^c \right) +l\alpha }$.
\begin{enumerate}
\item With probability $p$ attack at an element of $I$ chosen equiprobably at a start time chosen
uniformly at random in $J=\left[ 0,\alpha \right] $.

\item With probability $1-p$ attack uniformly on $W^c$ at start time $\alpha/2$.
\end{enumerate}
\end{definition}

The independent attack strategy randomizes over both time and space, unlike the strategy of the same name defined in \cite{AMP} for the discrete patrolling game, which randomizes only over space. The following result gives an upper bound on the strategy's interception probability.

\begin{theorem} \label{thm:independent}
Suppose $I$ is an independent subset of $Q$ of cardinality $l$. Then
\[
V\leq \frac{\alpha }{\lambda \left( W^c \right) +l\alpha },
\]%
which the Attacker can ensure by adopting the independent attack strategy.
If $\lambda \left( W^c \right) =0$ we have $V\leq 1/l$. Furthermore,
if $I$ is the set of leaf nodes, and leaf arcs have lengths exceeding $%
\alpha /2$, then
\[
V\leq \frac{\alpha }{\mu +l\alpha /2}.
\]
\end{theorem}
\begin{proof} Let $S$ denote any patrol and suppose the independent attack strategy is
adopted. If $S$ remains in $W$ during $J$, it intercepts the attack
with probability at most $p/l$, where $l$ is the cardinality of $I$.
Similarly, since $S$ has unit speed, if it remains in $W^c$ during
time $J$, it intercepts an attack with probability at most $(1-p)\left( \alpha
/\lambda \left( W^c \right) \right) $. The chosen value of $p$ is the
one that makes these probabilities both equal to $\alpha /\left( \lambda
\left( W^c \right) +l\alpha \right) $.

Finally, suppose the patrol $S$ starts in $W^c$ at time $0$, reaches a point $%
x\in I$  at some time $t$, $\alpha \leq t\leq 2\alpha $, early enough to
intercept some attacks on $I$ and late enough to intercept some attacks on $%
W^c$. Since the latest such a patrol can leave $W^c$ is at
time $t-\alpha /2$, it can cover a set of length at most $\left( t-\alpha
/2\right) -\left( \alpha /2\right) =t-\alpha $ in $W^c$ after the
attacks at time $\alpha /2$, intercepting a fraction $\left( t-\alpha
\right) /\lambda \left( W^c \right) $ of the attacks there. In
addition, the patrol can intercept the attacks at $x$ starting between $%
t-\alpha $ and $\alpha $, so a fraction $\left( 2\alpha -t\right) /\alpha $
of the attacks at $x$, or $\left( 2\alpha -t\right) /l \alpha$ of the
attacks on $I$. Thus the maximum probability that a patrol arriving at $I$
at time $t$ can intercept an attack is given by
\[
\left( 1-p\right) \frac{t-\alpha }{\lambda(W^c) }+p\frac{2\alpha -t}{l \alpha}=%
\frac{\alpha }{\lambda\left( W^c\right) + l \alpha. }
\]
By time symmetry, the same bound holds if the patrol starts at a point of $I$ and ends up in $W^c$.

If $\lambda \left( W^c \right) =0$ we have $V\leq 1/l$ trivially.

To prove the last assertion note that if $I$ is the set of leaf nodes, and leaf arcs have lengths exceeding $\alpha /2$, then leaf nodes form an independent set $I$ and $\lambda \left( W\right) =l\alpha /2$. 
\end{proof}

\subsection{A General Strategy Available to the Patroller}

Some patrol strategies come from finding closed paths on the network with
specific properties, and then have the Patroller go around them periodically
starting at a random point. Normally the closed path will be a tour, but we give a more general definition in case it is not.

\begin{definition}[Randomized periodic extension]
If $S:\left[ 0,L\right] \rightarrow Q$ is a closed unit speed path, we can
extend it to various patrols $S_{\Delta }:[0,\infty )\rightarrow Q$ of
period $L$ by the definition%
\[
S_{\Delta }\left( t\right) =S\left( ~\left( t+\Delta \right) \text{ mod }%
L~\right) ,\text{ for all }t\geq 0.
\]%
Thus $S_{\Delta }$ is a periodic patrol that starts at the point $S\left(
\Delta \right) $ at time $0$. The \textbf{randomized periodic extension} $%
\tilde{S}$ of $S$ is defined as the random mixture of the pure patrols $S_{\Delta },
$ with $\Delta $ chosen uniformly in the interval (or circle) $\left[ 0,L%
\right] $. In the special case that $S$ is a
Chinese Postman Tour, with $L=\bar{\mu}$, we call $\tilde{S}$ a Chinese
Postman Tour strategy.
\label{def:CPT}
\end{definition}

\subsection{$k-$covering Tours and Identifying Points of $Q$}

If a network $Q$ has an Eulerian tour, its randomized periodic extension
makes an effective patrolling strategy, because it visits all regular points
equally often (once), so the Attacker is indifferent as to where to attack.
If there is no Eulerian tour (the general case), we can still use this idea,
if there is a tour which visits all regular points equally often. In Theorem~\ref{thm:opt-tour} and Lemma~\ref{def:tree-patrol}, we will show that there is indeed such a tour which visits all regular points twice (a $2-$cover), with some additional properties. This idea is
formalized in the following.

\begin{theorem}
\label{k-tour}
Suppose $S:\left[ 0,L\right] \rightarrow Q$ is a closed unit speed tour that
visits every point of $Q$ at $k$ times which are separated by at least $%
\alpha $ (mod $L)$. Suppose $\tilde{S}$ is the randomized periodic extension of $%
S$ (from Definition \ref{def:CPT}). Then we have

\begin{description}
\item[(i)] $\tilde{S}$ intercepts any attack with probability at least $%
k\alpha /L$.

\item[(ii)] If $L=k\mu $, then the randomized periodic extension $\tilde{S}$ (for the Patroller) and a uniform
attack strategy (for the Attacker) are optimal and the value of the game is
given by $\alpha /\mu $.
\end{description}
\end{theorem}
\begin{proof} For part (i), suppose the attack takes place at a point $x$ in $Q$ starting
at some time $\tau $. Let $t_{i}$, $i=1,\dots ,k$ be times, separated by at
least $\alpha $, such that $S\left( t_{i}\right) =x$. The attack will be
intercepted by $S_{\Delta }$ if $\Delta $ is in the set $Y=\cup _{i}\left[
t_{i}-\tau -\alpha ,t_{i}-\tau \right]$ (modulo $L$), since in this case the Patroller
will visit $x=S\left( t_{i}\right) $ at some time in $\left[ \tau ,\tau
+\alpha \right] $. The separation assumption ensures that these intervals
are disjoint, and since they all have length $\alpha $, the length (Lebesgue
measure) of $Y$ is given by $\left\vert Y\right\vert =k\alpha $. By the
definition of $\tilde{S}$, the probability that $\Delta \in Y$ is equal to $%
\left\vert Y\right\vert /L=k\alpha /L$, as claimed in (i), so we have $V\geq
k\alpha /L=k\alpha /k\mu =\alpha /\mu $ under the assumption of part (ii).
By Lemma~\ref{lem:uniform}, we also have that $V\leq \alpha /\mu $, so the two inequalities
give $V=\alpha /\mu $, with $\tilde{S}$ and the uniform attack strategy
optimal. 
\end{proof}

As suggested above in the introductory remarks of this subsection, taking $k=1$ in Theorem~\ref{k-tour} gives another proof of the following elementary result
of \cite{ALMP} and \cite{Garrec}.

\begin{corollary}
If $Q$ is Eulerian, with Eulerian tour $S$, then for $\alpha \leq \mu $ we
have $V=\alpha /\mu $. ($V=1$ if $\alpha \geq \mu $.) In this case the randomized periodic extension $\tilde{S}$
and the uniform attack strategy are optimal for the Patroller and
Attacker, respectively. Furthermore, for a Chinese Postman Tour $S$ of any network $Q$, taking $k=1$
and $L=\bar{\mu}$ gives $V\geq \alpha /\bar{\mu}$.
\label{Eulerian}
\end{corollary}

It is useful to note for applications to patrolling by $m$ robots, that if
in Theorem 2 we require that $S$ visits every point at $k$ times separated
by time intervals $m\alpha $, then $m$ Patrollers can intercept any attack
with probability at least $mk\alpha /L$ (or $1$, if $mk\alpha /L \geq 1$). To see this conclusion, pick $\Delta $ as above and let the path of the $i$'th Patroller
(robot) be defined by $S_{i}\left( t\right) =S\left( \Delta +i(L/m)+t\right)
$. The arrival times at any point of $Q$ are then separated by at least $\alpha $. This reasoning shows that in our later lower bounds for $V$, these can be
multiplied by the number of Patrollers, with an upper bound of $1$.

We conclude this section with an observation on the effect of identifying
points of $Q$ on the value. \cite{AMP} considered the effect
of identifying two nodes of a graph. Here, we identify two {\em points} of the
network $Q$, using the well known quotient topology. In Figure~\ref{fig:identifying_points}
we identify the arc midpoints $C$ and $D$
of the network $Q$ to produce a new network~$Q^{\prime }$.

\begin{figure}[H]
\center
  \includegraphics[scale=0.3]{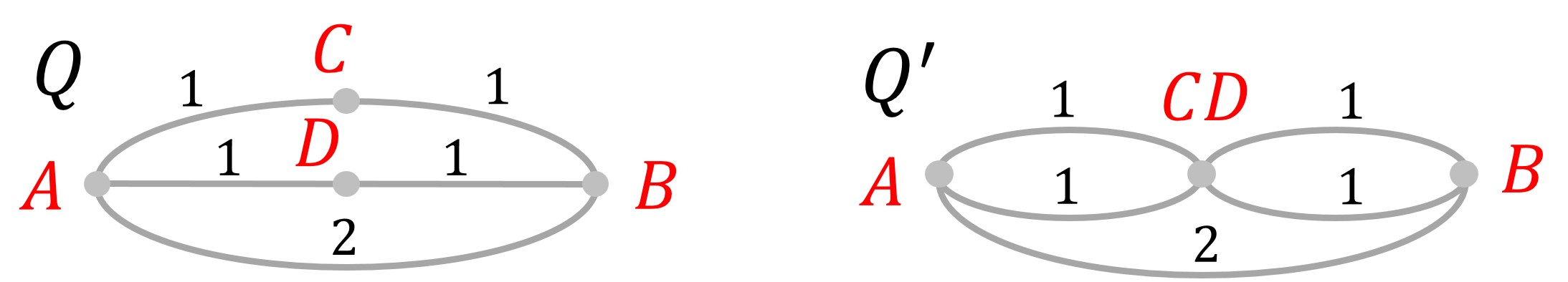}
  \caption{Identifying points $C$, $D$ of $Q$ to obtain $Q'$.}
  \label{fig:identifying_points}
\end{figure}

We may first look at two cases which have already been solved, the line
segment $Q_{line}=\left[ 0,1\right] $ and the circle $Q_{circle}=\left[ 0,1%
\right] ~\text{ mod }1$ (which is obtained from the line segment by identifying the endpoints), with say $\alpha =1/2$. From~\cite{ALMP}, we have $V\left( Q_{line}\right) =\alpha /\left( \mu +\alpha \right)
=1/3$. 
However as the circle is Eulerian, we have $V\left( Q_{circle}\right) =\alpha /\mu =1/2$, which is larger. It is easy
to show that identifying points cannot decrease the value. Of course if we
further identify points on the circle, we get new points of degree $4$, so
the resulting Eulerian network retains the value of $1/2$.

\begin{lemma} \label{lemma:identifying}
Suppose $Q^{\prime },d^{\prime }$ is the metric space obtained from $Q,d$ by
replacing the metric $d$ with a smaller metric $d^{\prime }$, that is, with $%
0\leq d^{\prime }\left( x,y\right) \leq d\left( x,y\right) $ for all $x,y\in
Q=Q^{\prime}$. Then $V\left( Q^{\prime },d^{\prime }\right) \geq V\left(
Q,d\right)$. Furthermore, if $Q^{\prime }$ is obtained from $Q$ by
decreasing the length of an arc or simply identifying two points $x$ and $y$, the
same result holds.
\end{lemma}

%

The proof of Lemma~\ref{lemma:identifying} is given in the Appendix. An application of it is given at the end of Section~\ref{sec:noleaves}.

\section{Networks Without Leaves} \label{sec:noleaves}

To extend Corollary~\ref{Eulerian} to general networks, we first note that
Eulerian networks have no leaf arcs, so we attempt to find such a
tour $S$ satisfying the hypothesis of Theorem~\ref{k-tour} for networks
without leaf arcs. It turns out that taking $k=2$ in Theorem~\ref{k-tour} is high enough.
We can find such a tour (see Theorem~\ref{theorem:network_noleaves}) if $\alpha $ is
sufficiently small with respect to the \emph{girth} $g$ of $Q$, defined for
networks as the minimum length of a circuit in $Q$, and if $Q$ has no circuits then $g = \infty$. (For networks with unit
length arcs, our definition of girth coincides with the usual integer definition of the girth of a
graph.) Our first result is the following.

\begin{theorem} \label{thm:opt-tour}
For any network $Q$ there is a tour $S_2$ which covers every arc twice and for which no arc is traversed
consecutively in opposite directions, except for leaf arcs.
\end{theorem}

Theorem~\ref{thm:opt-tour} is not new; it was proved by \cite{Sabidussi}. See also \cite{Klavzar} and \cite{Eggleton}. We originally proved Theorem~\ref{thm:opt-tour} independently and subsequently found it in the literature. Our proof, based on the new result, Lemma~\ref{lemma:paired}, is elementary.

The way we will prove Theorem~\ref{thm:opt-tour} is to double every arc of $Q$ to create an network $\hat{Q}$.
Then $\hat{Q}$ is Eulerian and has an Eulerian tour. We note that in Euler's Theorem (finding an Eulerian tour in
graphs of even degree), we can control to some extent the construction of
the tour. The following refinement of Euler's Theorem (Lemma~\ref{lemma:paired}) is based on some
simple modifications of the traditional proof and shows that we can control
the pairing of entered and exited \emph{passages} of the tour at every node.
Formally, a \emph{passage} at a node $x$ is a pair $\left( x,a\right) $, where $a$
is an arc incident to $x$. So a node of degree $d$ has $d$ passages and
every arc is part of two passages.

\begin{lemma} \label{lemma:paired}
Suppose $Q$ is a connected Eulerian network such that at every node the passages are identified in pairs (they are ``paired''). Then there is an Eulerian tour $S$ of $Q$ satisfying%
\begin{equation}
\text{$S$ never enters and leaves a node via paired passages.}
\label{paired}
\end{equation}
\end{lemma}

The proof of Lemma~\ref{lemma:paired} can be found in the Appendix.

As mentioned at the beginning of Section \ref{sec:definitions} there are no nodes of degree $2.$ Thus, the minimum node degree in our Eulerian network is $4.$

Now we are ready to prove Theorem~\ref{thm:opt-tour}.

\paragraph{Proof of Theorem \ref{thm:opt-tour}.}
Let $\hat{Q}$ be the Eulerian network obtained from $Q$ by doubling every arc. (This action has the effect of replacing leaf arcs with loops of double the length.)
At every node of $\hat{Q}$ we pair passages that correspond to the
same passage of $Q$. Now apply Lemma~\ref{lemma:paired} to $\hat{Q}$ to obtain an Eulerian
circuit $\hat{S}$ of $\hat{Q}$ satisfying condition (\ref{paired}). The result is $S_2$, a \emph{double cover of $Q$} (a tour of $Q$ where every arc is traversed twice), in which consecutive arcs are distinct, except for leaf arcs. For loops, an arc may be repeated consecutively, but always in the same direction both times.
\hfill$\Box$

The proof of Lemma~\ref{lemma:paired} gives rise to an algorithm for constructing an Eulerian tour of $\hat{Q}$ satisfying condition~(\ref{paired}), and hence a tour of $Q$ of the form described in the statement of Theorem~\ref{thm:opt-tour} (named $S_2$). Indeed, by following the rules listed in the proof of Lemma~\ref{lemma:paired}, we obtain a circuit $C$ in $\hat{Q}$ satisfying~(\ref{paired}); by recursively applying the rules to the connected components of $\hat{Q}-C$ and appending these circuits to $C$ at appropriate points, we can obtain an Eulerian tour of $\hat{Q}$ satisfying~(\ref{paired}).

We illustrate the creation of the $\ast$-circuit described above for the network $K_4$ depicted in Figure~\ref{fig:K4}. Doubling each arc, we give the extra arc the same label as the original arc but with a prime. Applying the rules of the proof of Lemma~\ref{lemma:paired}, starting at the bottom left node, we obtain a circuit: $a,b,c,d,e,c',a',f,d'$. Removing this circuit leaves the network consisting of arcs $b',e'$ and $f'$, which is already a circuit. Adding this circuit at the first possible opportunity, we obtain the Eulerian tour $a,b',e',f',b,c,d,e,c',a',f,d'$.

\begin{figure}[H]
\center
  \includegraphics[scale=0.3]{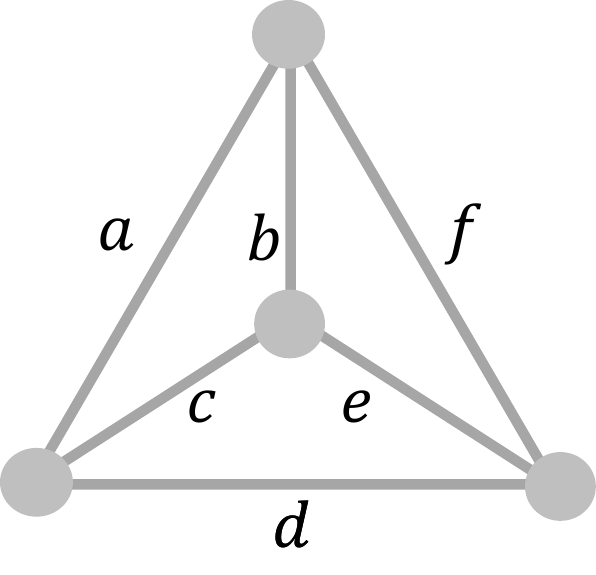}
  \caption{The network $K_4$.}
  \label{fig:K4}
\end{figure}

\begin{theorem} \label{theorem:network_noleaves}
Suppose $Q$ is a network without leaf arcs. Then for $\alpha \leq g$, where $g$ is the girth, we have the following:
\begin{enumerate}
\item	The value of the game is $V=\alpha/\mu$.
\item	For the Attacker, any uniform attack strategy is optimal.
\item	For the Patroller, the randomized periodic extension $\tilde{S_2}$ is optimal, for any tour $S_2$ given by Theorem~\ref{thm:opt-tour}.
\end{enumerate}
\end{theorem}
\begin{proof}
Let $S_2$ be a tour of $Q$ given by Theorem~\ref{thm:opt-tour}. Note that it has length $L = 2 \mu$. Since there are no leaf arcs,
any two consecutive arcs of $S_2$ are distinct. Suppose some point $x$ of $Q$
is reached by $S_2$ at consecutive times $t$ and $s$ with $t<s$. Let $Z$ denote the restriction
of $S_2$ to the interval $\left[ t,s\right]$. Then $Z$ is a circuit of length $s-t$ and hence $s-t\geq g$, by the definition of girth. Hence $V=\alpha/\mu$, by Theorem~\ref{k-tour}(ii) with $k=2$ and since $\alpha \leq g$. 
\end{proof}

For the network $K_4$ depicted in Figure~\ref{fig:K4}, assuming all arcs have length 1, the girth $g$ is $3$. So for $\alpha \leq 3$, the uniform attack strategy is optimal and the Patroller strategy $S_2$ is optimal, where $S_2$ is the tour $a,b',e',f',b,c,d,e,c',a',f,d'$.

As a further example, consider $Q$ to be a network with two nodes $A$ and $B$ connected by
three arcs of lengths $a \leq b \leq c$. Then $g=a+b$ and $\mu =a+b+c$, so we have by Theorem~\ref{theorem:network_noleaves} that the value is $V(\alpha)=\alpha /\left( a+b+c\right) $ for $\alpha \leq a+b$. This network, with $a=b=c=1$ (and hence $g=2$), was studied by \cite{Garrec}, who found (among other results) that $V(\alpha)=\alpha/3$ for $\alpha \le 2$ and $V(\alpha) \le f(\alpha) \equiv 1-(1/3)(2-\alpha/2)^2$ for $\alpha \in [2,10/3]$. Since $f(\alpha)<\alpha/3$ for $\alpha \in (2,10/3]$ ($f(\alpha)=\alpha/3$ for $a=2$ and $f'(\alpha)=(4-\alpha)/6<1/3$ for $\alpha>2$), the Patroller cannot obtain an interception probability of $\alpha/3$ for $\alpha$ in this interval, so the bound $\alpha \le g=2$ in Theorem~\ref{theorem:network_noleaves} is tight.

The condition $\alpha \le g$ specified in Theorem~\ref{theorem:network_noleaves} is a sufficient but not necessary condition. Consider a network $Q_{5}$ with two nodes connected by five arcs labeled as
$1,2,3,4,5$, with arc $i$ having length $i$. The girth is given by $g=g\left(
Q_{5}\right) =1+2=3$. However, suppose we obtain a double cover (with $k=2$) $%
S$ of $Q$  described by the sequence $[1,2^{\prime },3,4^{\prime
},5,1^{\prime },2,3^{\prime },4,5^{\prime}]$, where unprimed arcs go from, say, node $A$
to node $B$ and primed arcs go from node $B$ to node $A$. The shortest
return time to a regular point is for a point $x$ near node $B$ on the arc of length 5. After
leaving $x$, going to nearby $B$, the patrol traverses arcs of lengths $%
1+2+3+4=10$ before going back to $x$ from $B$. Note that $S$ returns to $A$ after gaps of $3,7,6,5$ and $9$, so at two time points separated by $14$ (at the start and after the gap of $6$). Also $B$ is visited twice separated by a gap of $14$.  So for the network $Q_{5}$
we have $V=\alpha /\mu $ for $\alpha \leq 10$ rather than just for $\alpha
\leq 3$. This observation leads to combinatorial questions about the maximum
shortest circuit in a $k$-cover of a network $Q$. As noted above based on
Garrec's analysis of the three arc network, in certain cases $V=\alpha /\mu $
fails for all $\alpha >g$.

Now let $Q$ be a network with two nodes connected by $n$ arcs.  If $n$ is even, then $Q$ is Eulerian and thus, by Corollary~\ref{Eulerian}, $V=\alpha/\mu$ for all $\alpha$. If $n$ is odd then our example $Q_{5}$ generalizes easily to the following.

\begin{theorem}
Suppose $Q$ is a network with two nodes connected by an odd number of arcs. Then $V=\alpha/\mu$
for $\alpha \leq \mu -D$, where $D$ is the length of the longest arc.
\end{theorem}

\begin{proof} Label the arcs between the two nodes $A$ and $B$ as $a_{1},\dots ,a_{n}$, in
order of increasing length $b_{1}\leq b_{2}\leq \dots \leq b_{n}$ where $b_{j}$ is the length of arc $a_{j}$ and $b_{n}=D$. We note that since the girth is given by $g=b_{1}+b_{2}$, Theorem~\ref{theorem:network_noleaves} says that $V=\alpha /\mu $ for $\alpha \leq g=b_{1}+b_{2}$. We have to establish the stronger result that $V=\alpha /\mu $ for $\alpha \leq b_{1}+b_{2}+\dots +b_{n-1}=\mu -D$. Following the construction of $S$ for $Q_{5}$ given above, we define a double tour $S$ of $Q.$ Let $j$ denote the traversal of arc $a_{j}$ from $A$ to $B$ and $j^{\prime }$ denote the traversal of arc $a_{j}$ from $B$ to $A$. Let $S$ be defined by the arc sequence $\left[ 1,2^{\prime },3,4^{\prime},\dots ,\left( n-2\right) ,\left( n-1\right) ^{\prime },n,1^{\prime},2,\dots ,n-1,n^{\prime }\right]$. Returns to any regular point $x$ of $Q$ occur after traversing $n-1$ of the arcs once. So the shortest return occurs when the arc not traversed is the longest one, namely arc $a_{n}$ of length $b_{n}=D.$ So the shortest return time under $S$ to any regular point is given by $b_{1}+b_{2}+\dots +b_{n-1}=\mu -D.$ So the double tour $S$ reaches every regular point $x$ twice at times separated by at least time $\mu -D$. So if $\alpha \leq \mu -D$ it reaches every regular point $x$ twice at times separated by at least time $\alpha$. Since the length of $S$ is given by $L=2\mu$, by Theorem~\ref{k-tour}, the value of the game is equal to $\alpha /\mu$.
\end{proof}

We conclude this section with an application of our earlier result on
identifying points.
\begin{example}
\textit{Consider the two networks $Q$ and $Q^{\prime }$ drawn in
Figure~\ref{fig:identifying_points}, with $\alpha =3$. We would like to show that $V\left( Q^{\prime
}\right) =\alpha /\mu =3/6=1/2$. We know from Lemma~\ref{lemma:identifying} that $V\left(
Q^{\prime }\right) \leq \alpha /\mu =1/2$. So we only need $1/2$ as a lower
bound on $V\left( Q^{\prime }\right)$. However we cannot apply Theorem~\ref{theorem:network_noleaves}
because it is not true that $\alpha$ is less than or equal to the girth of $Q'$, which is 2.
However we know either from \cite{Garrec} or from Theorem~\ref{theorem:network_noleaves} (which
applies because $3=\alpha <g=4$) that $V(Q) =\alpha /\mu =1/2$.
So by viewing $Q^{\prime }$ as coming from $Q$ by identifying points $C$ and
$D$, Lemma~\ref{lemma:identifying} gives $V\left( Q^{\prime }\right) \geq V\left(
Q\right) =1/2$.}
\end{example}

\section{Brief Attacks on Arbitrary Networks} \label{sec:arbitrary}

We now extend Theorem~\ref{theorem:network_noleaves} to networks with leaves. We begin with a modified
Patroller strategy based on the tour $S_2$ of Theorem~\ref{thm:opt-tour}.

\begin{definition}
Suppose $S_2$ is a tour given by Theorem~\ref{thm:opt-tour}. We denote by $S_2^{\alpha}$ the tour
that follows the same trajectory as $S_2$ but stops for time $\alpha $
whenever it reaches a leaf node.
\end{definition}

\begin{lemma} \label{lemma:LowerBound_leaves_girth}
Suppose $Q$ is a network with $l \ge 0$ leaf nodes and girth $g$. Then%
\[
V\geq \frac{\alpha }{\mu +l\alpha /2},\text{
for }\alpha \leq g.\text{ }
\]
\end{lemma}
\begin{proof}
Tour $S_2^{\alpha}$ takes total time $2\mu +l\alpha$. Note
that every point of $Q$ is visited by $S_2^{\alpha}$ at two times differing by
at least $\alpha$. So by Theorem~\ref{k-tour} part (i) with $k=2$, $L=2\mu +l\alpha$,
we have $V\geq 2\alpha /\left( 2\mu +l\alpha \right)$.
(We observe that instead of stopping for time $\alpha$, the tour $S_2^{\alpha}$
could do anything in this time interval, such as going away from the node a
distance $\alpha /2$ and returning.) 
\end{proof}

\begin{definition}[\textbf{Generalized girth}]
We define the \textbf{generalized girth} $g^{\ast }$ of a network $Q$ by
considering a leaf arc of length $L$ to be a circuit of length $2L$. So $%
g^{\ast }$ is the smaller between (1) the shortest circuit length of $Q$ and (2) twice the length of the shortest leaf arc. \label{def:gen_girth}
\end{definition}

In particular $g^{\ast }\leq \allowbreak g$, with equality if
there are no leaf arcs or if all leaf arcs have length greater than $g/2$. Note that if $\alpha \leq
g^{\ast }$ we know in particular that all leaf arcs have length at least $%
\alpha /2$ and hence Theorem~\ref{thm:independent} applies. Thus we have the following Attacker estimate (upper bound on $V$).

\begin{lemma} \label{lemma:generalised_girth}
Suppose $Q$ is a network with $l \ge 0$ leaf nodes and generalized girth $g^{\ast }$. Then
by adopting the independent attack strategy on the set $I$ of leaf nodes, the Attacker can ensure that the interception probability is less than $\frac{\alpha }{\mu +l\alpha /2}$ for $\alpha \leq g^*$. Hence,
\[
V\leq \frac{\alpha }{\mu +l\alpha /2},~\text{for }\alpha \leq g^{\ast }.
\]
\end{lemma}
\begin{proof}
As noted above, the assumption on $\alpha $ ensures that all leaf arcs have
length at least $\alpha /2$, so the result follows from Theorem~\ref{thm:independent}. 
\end{proof}

Since $g^{\ast }\leq g$, Lemmas \ref{lemma:LowerBound_leaves_girth} and \ref{lemma:generalised_girth} apply when $\alpha \leq g^{\ast }$ and
hence we have the following extension of Theorem~\ref{theorem:network_noleaves} to networks with leaf arcs.

\begin{theorem} \label{thm:gen-girth}
If $Q$ is a network with $l \ge 0$ leaf nodes and generalized girth $g^{\ast },$ then%
\[
V=\frac{\alpha }{\mu +l\alpha /2},~\text{for }\alpha \leq g^{\ast }.
\]%
For the Patroller, an optimal strategy is $S_2^{\alpha}$ as defined above. For the
Attacker, an optimal strategy is the independent attack strategy, taking $I
$ to be the independent set of leaf nodes.
\end{theorem}

Since $g^*$ is always positive, Theorem~\ref{thm:gen-girth} gives the solution of the game for some positive values of $\alpha$ on any network.

It is useful for later comparisons to specialize this result to trees.

\begin{corollary}
If $Q$ is a tree with $l$ leaf arcs, then%
\begin{enumerate}
\item[(i)] $V \geq \frac{\alpha}{\mu +l\alpha /2}$,
\item[(ii)] with equality if all leaf arcs have length at least $\alpha /2$.
\end{enumerate}
\label{cor:trees}
\end{corollary}
\begin{proof}
To establish (ii), note that trees have no circuits, so the
generalized girth $g^{\ast }$ is twice the length of its smallest leaf arc,
so by assumption,  $\alpha \leq g^{\ast}$. The result now follows from
Theorem~\ref{thm:gen-girth}.  For (i), consider the patrol $S_2^{\alpha}$.
Note that between any two visits by $S_2^{\alpha}$ to a
point of $Q$, a leaf node is visited. Hence the return times exceed the time
$\alpha $ that $S_2^{\alpha}$ stops at that node, and the result follows from
Theorem~\ref{k-tour}(i) with $k=2$ and $L=2\mu+l\alpha$. 
\end{proof}


For example, consider the tree $Q$ depicted in Figure~\ref{fig:dog-tree}. The number of leaf arcs is $l= 5$, the generalized girth is $g^*= 2$ and total length is $\mu= 9$, so by Theorem~\ref{thm:gen-girth}, the value of the game is $\alpha/(9 + 5\alpha/2)$ for $\alpha \le 2$. We will later solve the game for $\alpha \le 4$, using Theorem~\ref{theorem:trees}.

\begin{figure}
[H]
\center
  \includegraphics[scale=0.2]{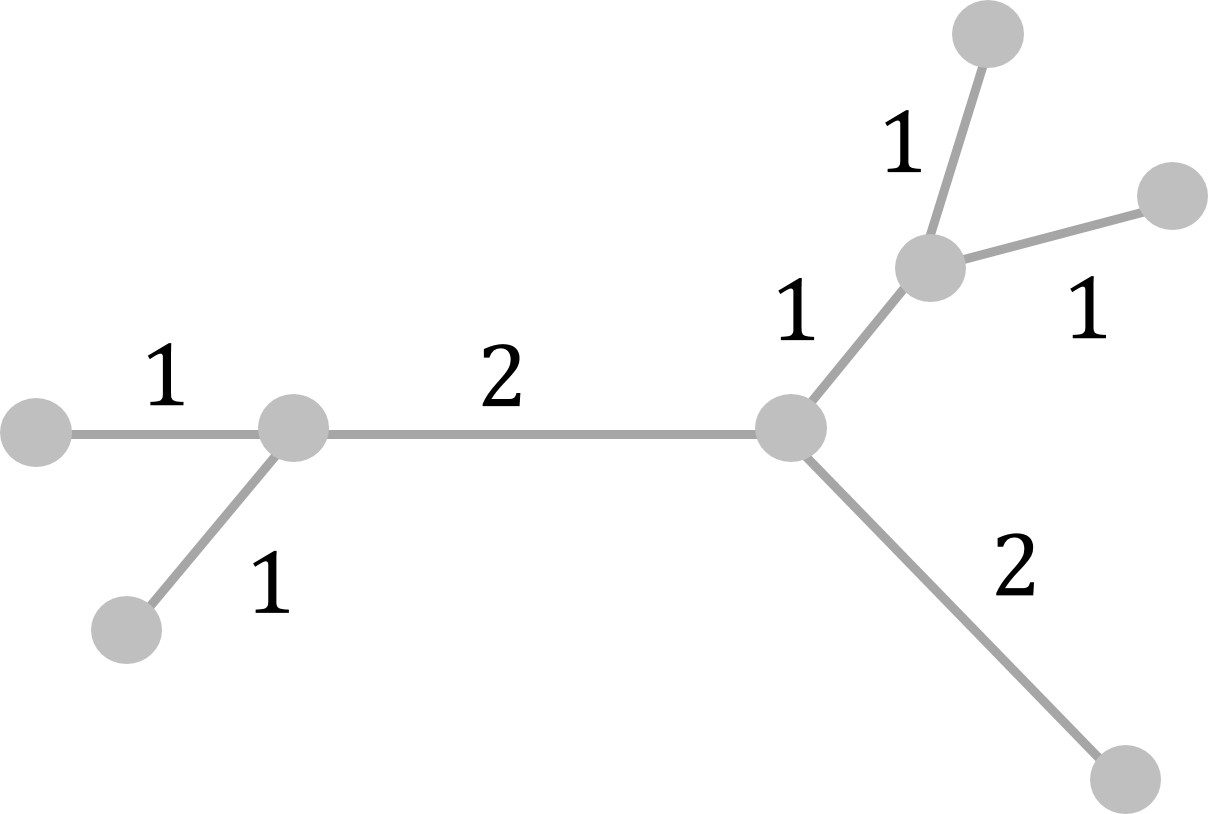}
  \caption{The tree $Q$.}
  \label{fig:dog-tree}
\end{figure}

\section{Solving the Game for Trees}
\label{sec:trees}

In Corollary~\ref{cor:trees} we gave some preliminary results for trees. Lemma~\ref{lemma:LowerBound_leaves_girth} gave a lower bound on the value of the game based on the Patroller strategy $S_2^{\alpha}$. Furthermore, for $\alpha \leq g^*$, where $g^*$ is the generalized girth, we showed in Theorem~\ref{thm:gen-girth} that the independent attack strategy ensures that this lower bound is tight. Note that for a tree, $g^*$ is twice the length of the shortest leaf arc. In this section, we extend these results and give optimal Patroller and Attacker strategies for some values of $\alpha$ which are greater than $g^*$. We start by defining the {\em extremity set} $E$, a subset of $Q$ that is essential in describing optimal Patroller and Attacker strategies.

\subsection{The Extremity Set $E$}
\label{sec:set_R}

The relationship between the network $Q$ and the duration $\alpha$ of the attack interval determines the type of optimal player strategies. In this section we define the extremity set $E$ that helps us explore this relationship for trees.

If $B$ is a set of points then we denote by $\bar{B}$ the topological closure of $B$. If $Q$ is a tree network, then its minimum tour time is $2 \mu$, as every arc must be traversed twice. If $x$ is a regular point of tree network $Q$, then $Q-\{x\}$ has two connected components $Q_1=Q_1(x)$ and $Q_2=Q_2(x)$, whose lengths satisfy $\lambda(Q_1) + \lambda(Q_2) = \lambda(Q)= \mu$. We introduce a subset $E$ of $Q$ called the {\em extremity set}.

\begin{definition}[\textbf{The extremity set $E$}]
Suppose $Q$ is a tree. The extremity set $E \equiv E(Q,\alpha)$ is defined as the set of all regular points $x \in Q$ such that
\begin{equation}
\min_{i=1,2} \lambda(Q_i(x)) < \alpha /2.
\label{eq:R_def}
\end{equation}
\end{definition}

Note that $\min_{i=1,2} \lambda(Q_i)  \leq \mu /2$ and if additionally $\mu < \alpha$ then (\ref{eq:R_def}) holds for all regular points, which implies that $\bar{E} = Q$. The extremity set $E$ consists of regular points whose minimum return time during a CPT is less than the attack duration $\alpha$. It can be partitioned into maximal connected sets that we call \emph{components} of $E$ and we denote by $E_j$.

\begin{example}
\textit{We illustrate the extremity set $E$ on the tree network of Figure~\ref{fig:dog-tree} that has $\mu=9$. Figure~\ref{fig:dog-tree-vs-alpha} shows how $E$ changes for increasing values of $\alpha$ on this network. As $\alpha$ increases the components grow starting from points near the five leaf nodes of the tree. Initially there are five components (cases $\alpha = 1,2,3,4$); but eventually points near non-leaf nodes become members of $E$ and the number of components increase to seven (cases $\alpha = 5,6,7,8$). Note that in case $\alpha = 8$ the closure $\bar{E}$ of $E$ is equal to the whole network. The results from the previous sections (Theorem~\ref{thm:gen-girth}, Corollary~\ref{cor:trees}) solve the game for cases $\alpha \leq g^* = 2,$ but in this section we extend the results to cover all cases of $\alpha \leq 4.$}

\begin{figure}[H]
    \centering
    \includegraphics[width = 4in]{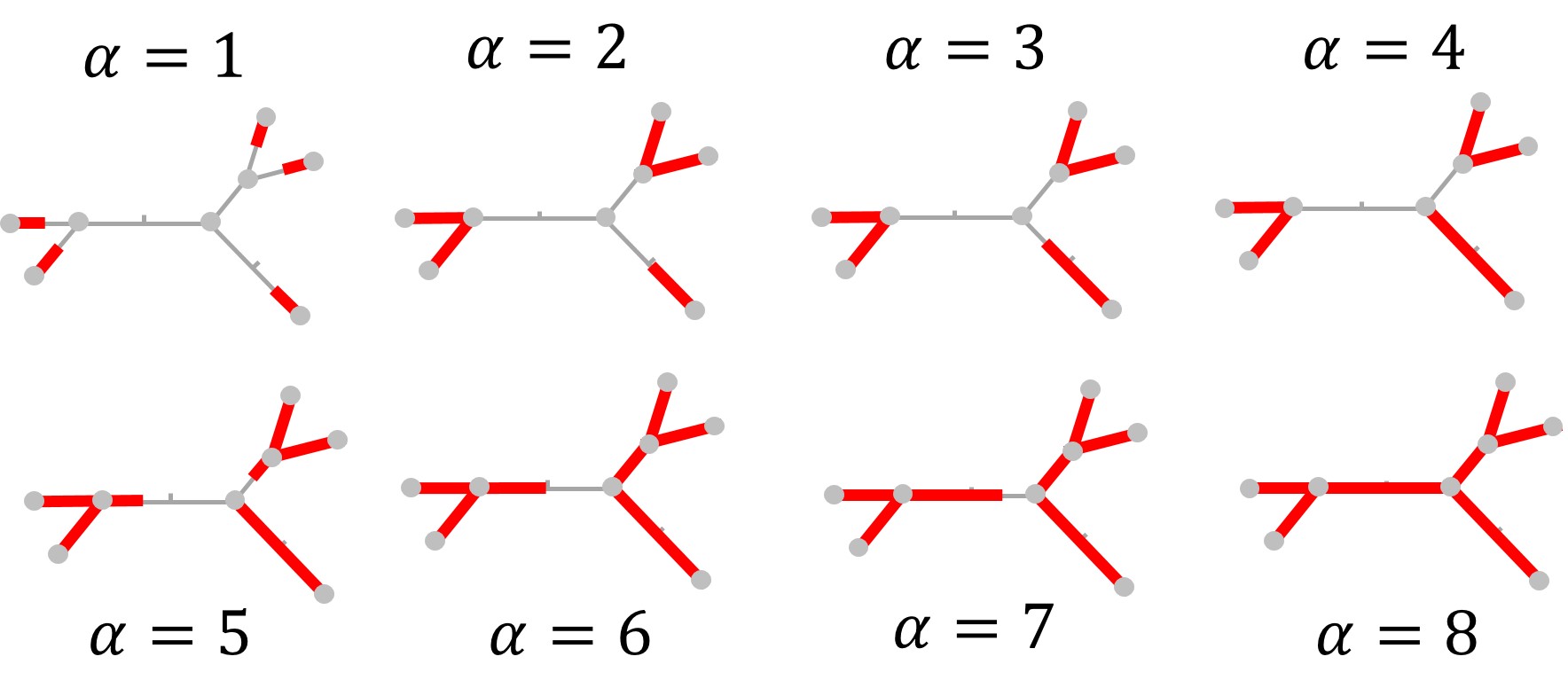}
    \caption{The extremity set $E(Q,\alpha)$, shown in thick (red) lines, for the tree $Q$ of Figure~\ref{fig:dog-tree} and $\alpha=1,\ldots,8$.}
    \label{fig:dog-tree-vs-alpha}
\end{figure}
\label{ex:dog_tree_vs_alpha}
\end{example}

\begin{example}
\textit{Figure~\ref{fig:tree_with_small_leaves} depicts a star network. The extremity set $E$ is depicted by red thick lines for attack time $\alpha$, if the lengths of $AD$ and $AF$ are each greater than $\alpha/2$ and those of $AB$ and $AC$ are each less than or equal to $\alpha/2$. Note that the nodes indicated by the small disks are not part of $E$.  Here, $E$ decomposes into four components: $(A,B)$, $(A,C)$, $(D,G)$, $(F,H)$. We claim that $\lambda(DG) = \lambda(FH) =  \alpha/2$; this is because on leaf arc AD (similarly for AF) if $\lambda(DG) < \alpha/2$ there would be a point X on the right of G whose distance from D would be $< \alpha/2$, implying $\lambda(DX) < \alpha/2$ and thus contradicting $X \notin E$. Similarly, if $\lambda(DG) > \alpha/2$ there would be a point X on the left of G where $\lambda(DX) > \alpha/2$ contradicting $X \in E$. Thus, components $E_j$ that are strict subsets of a leaf arc and whose closure contains the leaf node will have length $\alpha/2$. However, components $E_j$ whose closure is the entire leaf arc (like AB and AC) must have length $\leq \alpha/2$; if they had length $> \alpha/2$ then there would be point X on the component AB near node A where $\lambda(BX) > \alpha/2$ contradicting $X \in E$.}

\begin{figure}[H]
    \centering
    \includegraphics[width = 4in]{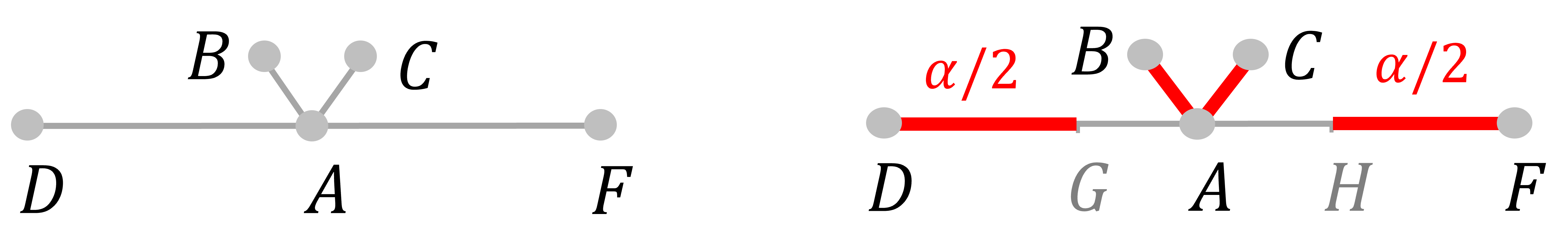}
    \caption{A tree, with its extremity set $E$ in thick red.}
    \label{fig:tree_with_small_leaves}
\end{figure}
\label{ex:tree_with_small_leaves}
\end{example}

\subsection{The $E$-patrolling Strategy $S^E$ for Trees}
\label{sec:trees_patroller}

We will see that for some trees, the uniform CPT strategy is still optimal for the Patroller, but its optimality depends on the size of the attack duration, $\alpha$. As mentioned earlier, for a tree a CPT is simply any depth-first search which returns to its start point after completing its search, so that $\bar{\mu}=2\mu$; every point of the tree except the leaf nodes is visited at least twice by a CPT. This means the leaf nodes and regular points near them are left ``less protected'' by a uniform CPT than the other points, and for sufficiently small values of $\alpha$, there will be points in the tree whose two closest visit times (modulo $\bar{\mu}$) are at least time $\alpha$ apart, meaning that they are, in a sense ``twice as protected'' as the leaf nodes. (In all that follows, arithmetic on time will be performed modulo the length of the tour in question).

This observation motivates the introduction of a new Patroller strategy $S^E$ for trees that we call the $E$-patrolling strategy. We construct it in such a way that each point is visited at least twice at times that differ by at least $\alpha$, and then we use Theorem~\ref{k-tour}, part (i) to obtain a lower bound on the value. To describe the strategy, we use the extremity set $E \equiv E(Q,\alpha)$ that we defined earlier; in particular, we use the closure $\bar{E}$ of $E$ and its components $\bar{E}^1,\ldots,\bar{E}^k$, each of which is a subtree of $Q$. We have $\lambda(\bar{E}) = \lambda(E)$ but by using the components of $\bar{E}$ rather than the components of $E$, we include the nodes and thereby unite adjacent components of $E$ into a single component of $\bar{E}$. For example, in Figure~\ref{fig:tree_with_small_leaves} there are four components of $E$ but only three components of $\bar{E}$, since in $\bar{E}$ the lines AB and AC join to form a single component BAC.

Let $Q$ be a tree with $\bar{E}\neq Q$. We first construct a CPT $S$ with
the additional property that every component $\bar{E}^{j}$ is searched in a single CPT of $\bar{E}^{j}$, which we call $C_j$; note that some CPTs of $Q$ might search different subsets of $\bar{E}^{j}$ during non-consecutive time intervals - we exclude this possibility by construction.

To obtain a CPT of $Q$ with this property, we begin at any regular point not in $\bar{E}$ and go in either direction.
When arriving at any node, we leave by a passage not already traversed, if
there is such a passage. (This is the usual depth-first construction and
ensures we obtain a CPT.) Furthermore, if the node belongs to some component
$\bar{E}^{j}$ and there are untraversed passages staying in that component,
we take one of these. For example, in Figure~\ref{fig:tree_with_small_leaves} if we start somewhere on $GA$ going
right,  and tour the leaf arc to $B$ from $A$, we must then take the passage to $%
C$ (staying in component BAC) rather than the other untraversed passage out
of $A$ going to $F$. This rule ensures that the CPT say $ABAFACADA$ (in which the component $BAC$ of $\bar{E}$ is not traversed in a single CPT of $BAC$) will not be
constructed, but rather one like $[ABACA]FADA$, where the bracketed
expression is a CPT of the component $BAC$.

Then we make two types of additions at every component. If $\lambda \left(
\bar{E}^j\right) \geq \alpha /2$, we follow the  CPT $C_{j}$ of $\bar{E}^{j}$ in $S$ by
another identical one, before continuing with $S$. Note that this local CPT
takes time $\geq \alpha $, so the time between the first and second CPT of $%
\bar{E}^j$ reaching any (regular) point is at least $\alpha $.


If $\lambda \left( \bar{E}^j\right) <\alpha /2$ we wait until $S$ comes back to $\bar{E}^j$
after the first occurrence of $C_{j}$ in $S,$ and then insert a second $C_{j}.$ Let $[t_1,t_2]$ be the time interval during which $S$ tours $\bar{E}^j$ so that $S(t_1)=S(t_2)$ and $t_2-t_1=2\lambda(\bar{E}^j).$ We have $\alpha > t_2-t_1.$ In this case, we cannot simply tour $\bar{E}^j$ twice in succession, because some points in $\bar{E}^j$ will not be visited at two times that are at least time $\alpha$ apart. Let $x = S(t_1) = S(t_2)$, and we claim that $x$ is a (non-leaf) node of the network. For suppose not, so that $x$ is a regular point, and let $x' \notin E$ be on the same arc with $\varepsilon:=d(x,x') < \alpha/2 - \lambda(\bar{E}^j)$. Then the length of $S\left( (t_1-\varepsilon, t_2+\varepsilon) \right),$ which is  $\lambda(\bar{E}^j) + \varepsilon,$ is less than $\alpha/2.$ The set $S\left( (t_1-\varepsilon, t_2+\varepsilon) \right)$ is a component of $Q-x’$ and by the definition of $E,$ since the smaller component of $Q-x'$ has length less than $\alpha/2,$ we have $x' \in E,$ a contradiction. So $x$ is a non-leaf node, and thus $Q-x$ has at least three components. If any component $A$ of $Q-x$ has length less than $\alpha/2,$ then its closure $\bar{A},$ which contains $x,$ must be a subset of $\bar{E},$ and hence of $\bar{E^j}$ (since $x \in \bar{E^j}$). Hence, all components of $Q-x$ that are disjoint from $\bar{E^j}$ must have length at least $\alpha/2.$
So the next time after $t_2$ that $S$ arrives at $x$ is $t_3 \ge t_2 + \alpha$, and the next time after $t_3$ that $S$ arrives at $x$ is at least $t_3+\alpha$. Then $S$ is updated by adding another tour of $C_j$ at time $t_3$.

Observe that each additional local CPT of $\bar{E}^j$ takes time $2\lambda
\left( \bar{E}^j\right) $, so the total length of the resulting tour $S^E$ is $2\mu +2\left( \sum_j \lambda \left( \bar{E}^j\right)
\right) =2\left( \mu + \lambda\left(\bar{E}\right)\right) $ and by construction it reaches every
point of $Q$ at two times separated by at least $\alpha$ (modulo the length of the tour). Note that if $\bar{E}=Q$,
we simply take $S^E=S$. The optimal periodic strategy is thus $S^E$. For the network of Figure~\ref{fig:tree_with_small_leaves}, taking $S$ as $ABACADAFA$ we could have $S^E=ABACAGD\:[GDG]\:[ABACA]\:HF\:[HFH]\:A$, where the brackets
indicate the three inserted local CPT's of the components of $\bar{E}$. Note
that two of these are inserted right after their first occurrence, but the
third one [ABACA] is inserted nonconsecutively. Our construction would not
work directly on the CPT $ABAFACADA$.

Thus we have established the following result by explicit construction.

\begin{lemma}
	Suppose $Q$ is a tree. Then there is a tour $S^E$, called an {\em $E$-patrolling strategy}, of length $2\left(\mu+ \lambda(E)\right)$ such that every point $x$ of $Q$ is visited at least twice at times that differ by at least~$\alpha$.
	\label{def:tree-patrol}
\end{lemma}

We can obtain a lower bound on the value of the game obtained by using an $E$-patrolling strategy.

\begin{lemma} 	\label{lem:patroller_bound_trees}
	Suppose $Q$ is a tree. Any $E$-patrolling strategy intercepts any attack with probability at least $v^* \equiv \alpha / (\mu + \lambda(E))$.
\label{thm:R-patrolling}
\end{lemma}
\begin{proof}
Follows from Lemma~\ref{def:tree-patrol} and Theorem~\ref{k-tour} part (i) with $k=2$, $S=S^E$, and $L=2\left( \mu +\lambda(E) \right)$. 
%
\end{proof}

We conjecture the following on trees:

\begin{conj} \label{conj} If $Q$ is a tree network, then for any $\alpha$ the $E$-patrolling strategy is optimal and the value of the game is $V=v^*\equiv \alpha/ (\mu + \lambda(E))$.
\end{conj}

We later confirm the conjecture in some special cases.

Note that when $\alpha \leq g^*$ we have $\lambda(E) = l \alpha/2$, and the result of Lemma~\ref{thm:R-patrolling} becomes the same as the result of Corollary~\ref{cor:trees}. In that case, the patrolling strategy $S_2^{\alpha}$ gives the same lower bound as an $E$-patrolling strategy.

\subsection{The $E$-attack Strategy}
\label{sec:trees_attacker}

In the previous section we showed that on a tree, any $E$-patrolling strategy intercepts any attack with probability at least $v^{\ast }$.
Here, we define the {\em $E$-attack strategy}, whose attacks are intercepted with probability at most $v^{\ast }$ on some trees.
The condition that allows
this strategy to be defined and to be optimal is given in Definition~\ref{def:leaf-cond}. It is useful to note that while for patrolling strategies we looked at the components of the closure $\bar{E}$ of $E$, for the attack strategy given here we look at the components of $E$ itself.

\begin{definition}[\textbf{Leaf Condition}] \label{def:leaf-cond}
Suppose $Q$ is a tree.
We say that $(Q,\alpha) $ satisfies the \textbf{Leaf Condition} if the extremity set $E$ consists of all points on every leaf arc within distance $\alpha /2$ of its leaf node.
\end{definition}

For example, in Figure~\ref{fig:dog-tree-vs-alpha} the cases that satisfy the Leaf Condition are the first four ($\alpha =1,2,3,4$), where $E$ consist of five components; all of these five components are subsets of leaf arcs and they are within $\alpha/2$ from the leaf node. Note that the Leaf Condition implies that every component $E_j$ of $E$ corresponds to a leaf node; this is easy to check in Figure~\ref{fig:dog-tree-vs-alpha}. Cases $\alpha=5,6,7,8$ have seven components; five of these components are subsets of leaf arcs but two of them are subsets of non-leaf arcs and thus $(Q,\alpha)$ does not satisfy the Leaf Condition. (Recall that the extremity set does not contain nodes, thus the nodes separate the components.)

\begin{definition} [\textbf{$E$-attack strategy}]
\label{def:spread_strategy}
Suppose $(Q,\alpha)$ satisfies the Leaf Condition, where $Q$ is a tree. Let $x_{j}$ denote
the leaf node contained in the closure of the component $E_{j}$ of $E$, and let $e_j=\lambda ( E_j) $ and let $M=\max_j \lambda (E_j)  $ be the maximum length of a component of $E$.
We define the \textbf{$E$-attack strategy} as follows:

\begin{enumerate}
\item With probability $\lambda(E^c)/( \mu+\lambda(E))$, attack a uniformly random point of $E^c$ at time $M$.

\item With probability $2e_j/(\mu +\lambda(E))$, attack at leaf node $x_{j}$ at a start time chosen uniformly in the interval $[M-e_j,M+e_j]$.
\end{enumerate}
\end{definition}
Note that the Leaf Condition implies that $\sum_j e_j = \lambda(E)$, therefore the sum of the probabilities from 1. and 2. above sum to $1$. Also, unlike the uniform attack strategy, the $E$-attack strategy is not synchronous. That is, the attack does not start at a fixed, deterministic time.

\begin{example}
\textit{We revisit Figure~\ref{fig:tree_with_small_leaves}, where the leaf arcs have lengths $2,1,6,6$ and $\alpha = 6$. We illustrate the $E$-attack strategy on this star network in Figure~\ref{fig:star_R-attack-strategy2}. Here $\mu=15$; the extremity set $E$ is shown in thick red lines. $E$ consists of four components that are subsets of leaf arcs and whose points are within $\alpha/2$ from the leaf node, thus the Leaf Condition is satisfied. Also, note that $\lambda(E) = 9$ and $\mu+\lambda(E) = 24$. The $E$-attack strategy then attacks as follows: with equal probabilities $6/24$ it attacks at nodes $D$ and $F$ with a starting time chosen uniformly on $[0,6]$; with probabilities $4/24$, $2/24$ it attacks leaf nodes $B$, $C$ with a starting time chosen uniformly on $[1,5]$, $[2,4]$ respectively; with probability $6/24$ it attacks uniformly on set $E^c$ at time $M=3$.}

\begin{figure}[H]
    \centering
    \includegraphics[width = 4.5in]{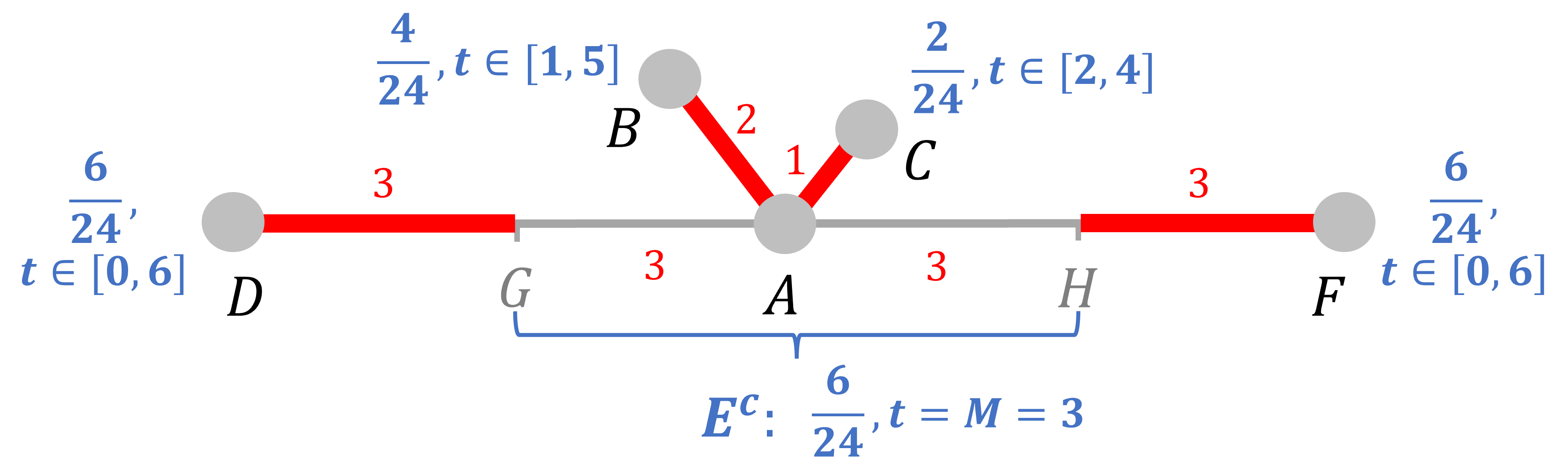}
    \caption{The $E$-attack strategy on an asymmetric star with arcs lengths 2,1,6,6 with $\alpha=6$. The set $E$ is shown in thick red lines.} 
    \label{fig:star_R-attack-strategy2}
\end{figure}
\label{ex:star_R-attack-strategy2}
\end{example}

We next prove that for trees $Q$, the $E$-attack strategy is optimal if $(Q,\alpha)$ satisfies the Leaf Condition.

\begin{lemma}
Suppose $Q$ is a tree and $(Q,\alpha)$ satisfies the Leaf Condition. Then the $E$-attack strategy is intercepted by any patrol with probability at most $v^*= \alpha/ (\mu + \lambda(E))$.
\label{lemma:spread_trees}
\end{lemma}

The proof of Lemma \ref{lemma:spread_trees} is in the Appendix. If we combine the results of Lemma~\ref{thm:R-patrolling} and Lemma~\ref{lemma:spread_trees} on patrolling and attack strategies for trees, we obtain the following exact result for the value of the game.

\begin{theorem} \label{theorem:trees}
Suppose $Q$ is a tree and $(Q,\alpha)$ satisfies the Leaf Condition. Then any $E$-patrolling strategy is optimal, the $E$-attack strategy is optimal, and the value of the game is~$V = v^*$.
\end{theorem}

\begin{example}
\textit{We revisit the network $Q$ from Figure~\ref{fig:star_R-attack-strategy2} with $\alpha=6$ and $\mu=15$. We first consider patrolling strategies. The $S_2^{\alpha}$ patrolling strategy is $ADDABBACCAFFA$, where repeating a node means it stays there for duration $\alpha$; this tour has length $2 \mu + 4 (6) = 54$. From Corollary~\ref{cor:trees} we have $V \geq \frac{\alpha}{\mu +l\alpha /2} = 6/27$. An $E$-patrolling strategy is $ADGDABACABACAFHFA$ with length $2 \mu + 2 \lambda(E) =48$; from Lemma~\ref{thm:R-patrolling} we have $V \geq v^* = \frac{\alpha}{\mu + \lambda(E)} = 6/24$. As we can see, an $E$-patrolling strategy, which is defined only for trees offers an improvement over the $S_2^{\alpha}$ patrolling strategy, which is a more general strategy.}

\textit{Now, we consider attacker strategies. Let $I$ be the set of leaf nodes. The sets $E$ and $W \equiv W(I)$ are shown in Figure~\ref{fig:star_with_W_R} with solid thick red and dashed thick green lines respectively. Note that $(Q,\alpha)$ satisfies the Leaf Condition. The $E$-attack strategy is demonstrated in Figure~\ref{fig:star_R-attack-strategy2} and it gives a lower bound, $v^* = \frac{\alpha}{\mu + \lambda(E)} = 6/24$, from Theorem~\ref{theorem:trees}, which is optimal. The bound given by Theorem~\ref{thm:independent} $\frac{\alpha}{\lambda(W^c)+ l \alpha} = \frac{\alpha}{\mu+ l \alpha/2} = 6/27$ does not hold in this case because $I$ is not an independent set or, equivalently, leaf arcs do not have lengths exceeding $\alpha/2$.}

\begin{figure}[H]
    \centering
    \includegraphics[width = 3in]{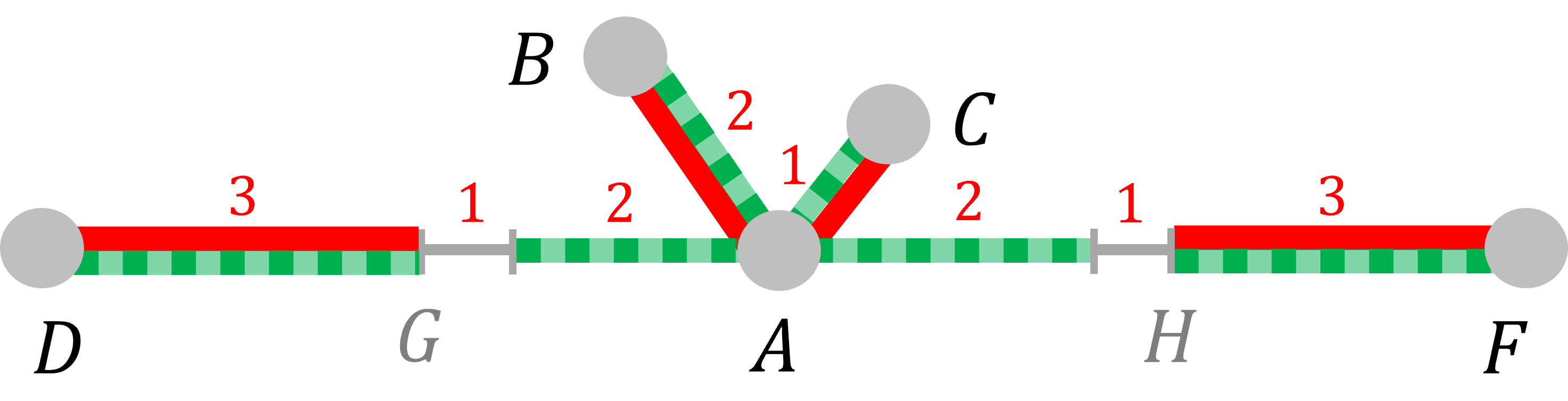}
    \caption{Star with arc lengths 6,6,2,1 and $\alpha = 6$. The solid thick red line is the set $E$ and the thick dashed green line is the set $W \equiv W(I)$, where $I$ is the set of leaf nodes; note that here $I$ is not an independent set.}
    \label{fig:star_with_W_R}
\end{figure}

\label{ex:R-patroller-strategy}
\end{example}

A star is a network consisting entirely of leaf arcs. We call a star {\em balanced} if no arc comprises more than half of its total length; otherwise we say that it is {\em skewed}. It is easy to check that balanced stars satisfy the Leaf Condition. All symmetric stars (whose arcs are all the same length) are balanced. An example of a skewed star is a star with $n$ arcs of length $1$ and one arc of length $x > n$, as shown in Figure~\ref{fig:x-n-star}; the long arc has length $x,$ which is more than half of $\mu = n + x.$

It is also easy to see that if $Q$ is a star (which may be balanced or skewed) whose longest arc has length at most $\alpha/2$, then $\bar{E}=Q$ and hence $Q$ satisfies the Leaf Condition. So Theorem~\ref{theorem:trees} gives the following.

\begin{corollary} \label{cor:stars}
Suppose $Q$ is a star.  Then the $E$-attack strategy and any $E$-patrolling strategy are optimal and the value of the game is $V=v^*= \alpha/ (\mu + \lambda(E))$ if either
\begin{enumerate}
\item[(i)] $Q$ is balanced or
\item[(ii)] $\alpha$ is at least twice the length of the longest arc of $Q$.
\end{enumerate}
\end{corollary}

Note that if $Q$ is the line segment network, then by adding an artificial node in the center, we can apply Corollary~\ref{cor:stars}, part (i), recovering the result for the value of this game, given previously in~\cite{ALMP} (though the optimal strategies given here are different).

\subsection{Stars Not Satisfying the Leaf Condition} \label{sec:trees2}

In Lemma~\ref{thm:R-patrolling} we showed that the $E$-patrolling strategy intercepts any attack with
probability at least $v^{\ast }=\alpha /\left( \mu +\lambda \left( E\right)\right) $ and that (Lemma~\ref{lemma:spread_trees}) for trees satisfying the Leaf Condition, the $E$-attack strategy avoids interception with probability at least $v^{\ast }$.
Thus for trees we have $V=v^{\ast }$ if the Leaf Condition is satisfied, but what happens when it is not satisfied?
In this subsection we present a class of trees $Q$ for which the Leaf Condition fails for some values of $\alpha$ but nevertheless $V=v^{\ast }$ for all values of $\alpha$.
We do this by specifying particular attack strategies which are optimal on these trees.

We consider the class of skewed stars with $n$ arcs of length $1$ and one arc of length $x >n$, as shown in Figure~\ref{fig:x-n-star}. We refer to these skewed stars as {\em symmetric skewed stars}. The degree $1$ nodes incident to the arcs of length $1$ are denoted $a_1,\ldots,a_n$, the node of degree $n+1$ is denoted $a_0$ and the degree $1$ node at the end of the arc of length $x$ is denoted $b$. It is easy to see that symmetric skewed stars satisfy the Leaf Condition only for $\alpha \leq 2n$ and $\alpha \geq 2x$. In what follows we introduce attack strategies for these stars that guarantee $v^*$ for the attacker for $2n \leq \alpha \leq 2x$, and thus show that Conjecture~\ref{conj} holds for symmetric skewed stars for all values of $\alpha$.  Later, in Subsection~\ref{sec:simple_tree} we give an attack strategy on a non-star tree that also guarantees the value $v^*$ for the attacker and show that Conjecture~\ref{conj} holds for this example.

\begin{figure}[!ht]
    \centering
     \begin{tikzpicture}[line width=2pt]
    \tikzstyle{every node}=[draw,circle,fill=black,minimum size=1.5pt, inner sep=0pt]
    \draw (0,0) node (u)[label=90:${\ u\ }$, minimum size=1.5pt, color = black] {}
        -- ++(-180:1cm) [color= red,line width=3pt ] node (3) [label=60:${\ a_0\ }$, minimum size=4pt, color= black] {}
        -- (3);
    \draw (u)
        -- ++(0:5cm) node (v) [label=90:${\ v\ }$, minimum size=1.5pt, color= black] {}
        -- (u);
    \draw (v)
        -- ++(0:7cm) [color= red,line width=4pt ]node (x) [label=90:${\ b\ }$, minimum size=4pt, color= black] {}
        -- (v);
    \draw (3)
     -- ++(120:1.5cm) [color=red, line width=4pt ] node (1) [label=90:${\ a_1\ }$, minimum size=4pt, color=black]{}
     -- (3);
    \draw (3)
     -- ++(150:1.5cm) [color=red, line width=4pt ] node (2) [ minimum size=4pt, color=black]{}
     -- (3);
     \draw (3)
     -- ++(180:1.5cm) [color=red, line width=4pt ] node (6) [ minimum size=4pt, color=black]{}
     -- (3);
     \draw (3)
     -- ++(-150:1.5cm) [color=red, line width=4pt ] node (4) [minimum size=4pt, color=black]{}
     -- (3);
    \draw (3)
     -- ++(-120:1.5cm) [color=red, line width=4pt ] node (5)[label=-60:${\ a_n\ }$, minimum size=4pt, color= black] {}
     -- (3);
\end{tikzpicture}
    \caption{A symmetric skewed star. The extremity set $E$ consists of the $n+2$ thick (red) lines. The black line is the set $E^c$.}
    \label{fig:x-n-star}
\end{figure}


We define an attack strategy that we will show is optimal for symmetric skewed stars for $2n \leq \alpha \leq 2x$. 
We note that for the a symmetric skewed star with $2n \leq \alpha \leq 2x$ it is easy to check that $\lambda(E)=\alpha$ if $2n\leq \alpha \leq \mu=x+n$ and $\lambda(E)=\mu$ (equivalently, $\lambda(E^c)=0$) if $\alpha \geq \mu$. We denote the left and right boundary points of $E^c$ with $E$ by $u$ and $v$ respectively; since $2n \leq \alpha \leq 2x$, both of these points are on the long arc or on its boundary.

We note that the Leaf Condition for this star holds for $\alpha=2n$ but not for $2n < \alpha \leq 2x$, thus the $E$-attack strategy is not defined for the latter set of values. Thus, we define a new attack strategy. For $\alpha=2n$ either strategy can be used.

\begin{definition}[\textbf{Symmetric-skewed attack}]
The symmetric-skewed attack strategy is defined as follows:

\textit{Left attacks: }With probability $(2\lambda(E)-\alpha)/(\mu+\lambda(E))$, attack equiprobably at nodes $a_i$, for $i=1,...,n$, starting uniformly at times in $[n-1, \alpha+n-1]$ if $2n\leq \alpha \leq x+n$ and at times in $[\alpha-x-1, x+2(n-1)+1]$ if $x+n \leq \alpha \leq 2x$.

\textit{Middle attacks: }With probability $\lambda(E^c)/(\mu+\lambda(E))$, attack at a uniformly random point of $E^c$, starting equiprobably at times $\alpha/2+2j$ for $j =0,1,..., n-1$.

\textit{Right attacks: }With probability $\alpha/(\mu+\lambda(E))$, attack node $b$, starting at a time in $[0, \alpha+2(n-1)]$ chosen as follows: conditional on the attack taking place here, the starting time is given by the following probability cumulative function. For $z=1,..., n-1$,
\begin{equation*}
    f(y)=
    \begin{cases}
    \frac{z(y-z+1)}{n\alpha} & \text{if}\ \ 2(z-1)\leq y\leq 2z,\\
    \frac{n-1}{\alpha}+\frac{y-2(n-1)}{\alpha} & \text{if}\ \ 2(n-1)\leq y\leq \alpha,\\
    \frac{\alpha-(n-1)}{\alpha}+\frac{z(z-1)}{n\alpha}+\frac{(y-\alpha)(n-z)}{n\alpha} & \text{if}\ \ \alpha+2(z-1)\leq y\leq \alpha+2z.\\
    \end{cases}
\end{equation*}
\end{definition}
Note that when $\alpha \geq \mu$, we have $\lambda(E^c)=0$ so there are no middle attacks.
\begin{theorem} \label{thm:x-n-star}
Suppose $Q$ is a symmetric skewed star. For any $\alpha$ the $E$-patrolling is optimal and the value of the game is $V=v^*=\alpha/(\mu+\lambda(E))$. If $x>n$ and $2n\leq \alpha \leq 2x$ then the symmetric-skewed attack strategy is optimal, otherwise the $E$-attack strategy is optimal.
\end{theorem}

The proof of Theorem~\ref{thm:x-n-star} can be found in the Appendix. Theorem~\ref{thm:x-n-star} provides a counterexample to a conjecture in \cite{ALMP}. The conjecture was that for trees, if $\alpha$ is at least the diameter of the network, the value of the game is $\alpha/\bar{\mu} = \alpha/(2 \mu)$. For a symmetric skewed star, the diameter is $x+1$, and by Theorem~\ref{thm:x-n-star}, for $x+1 \le \alpha < 2x$, the value is $\alpha/(\mu + \lambda(E))$. This is not equal to $\alpha/(2 \mu)$, since $\lambda(E) < \mu$ in that range of $\alpha$, disproving the conjecture in \cite{ALMP}.


\subsection{A non-star tree with $\bar{E} = Q$ satisfying Conjecture~\ref{conj}.}
\label{sec:simple_tree}

We now consider the tree depicted in Figure~\ref{fig:simple_tree} with unit length arcs and $\alpha = 6$. This gives $\bar{E}=Q$ and thus $\lambda(E) = \mu$. Here $\mu=6$ and thus $v^* = \alpha/2\mu = 1/2$.

\begin{figure}[h]
    \centering
    \includegraphics[width = 1.3in]{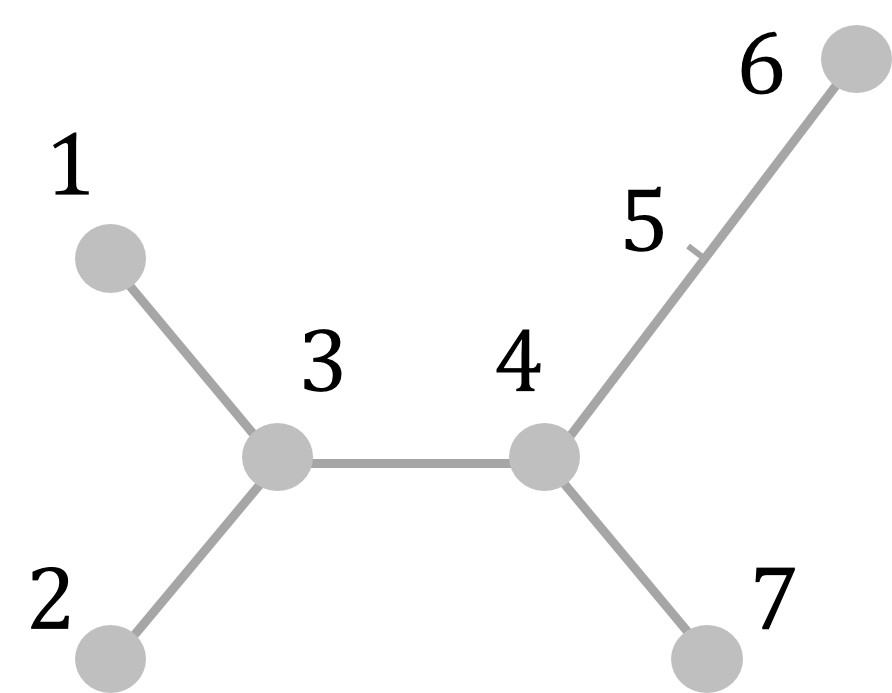}
    \caption{A tree with $\mu=6$.}
    \label{fig:simple_tree}
\end{figure}

We propose the following Attacker strategy for this specific tree with $\alpha = 6$.

\begin{itemize}
\item At each leaf node $1$ and $2$ attack with probability $6/24$ at a start time chosen uniformly in the interval $[0,6]$ (total attack probability $12/24$).
\item At leaf node $6$ attack with attack start time uniformly: in the interval $[0,2]$ with probability $2/24$, in the interval $[2,4]$ with probability $4/24$, in the interval $[4,6]$ with probability $2/24$ (total attack probability $8/24$).
\item At leaf node $7$ attack takes place with probability $4/24$ at a start time chosen uniformly in the interval $[1,5]$ (total attack probability $4/24$).
\end{itemize}

It is easy to verify that the probability of interception guaranteed by this strategy is $v^* = 1/2$, thus showing that the conjecture holds for this example; the proof is along the same lines as that of Theorem~\ref{thm:x-n-star}.

\section{Conclusions}

This paper models the problem of patrolling a pipeline or road system
against attacks which can be made anywhere, not just at a discrete set of
``targets''. We do this by analyzing the continuous patrolling game on
arbitrary metric networks $Q,d,$ where $d$ is the shortest path metric. The
Attacker picks a point of $Q$ to attack (not necessarily a node) during a
chosen time interval of given length $\alpha .$ The Patroller chooses a unit
speed path in the network and wins the game if the path crosses the attacked
point during the attack; otherwise the Attacker wins.  Mixed strategies are
required for optimal play in this game, where the payoff to the maximizing
Attacker is the probability that the attack is intercepted. Prior work of
\cite{ALMP} and \cite{Garrec} has solved the game for Eulerian
networks, the line (or interval) network and a network consisting of two
nodes connected by three arcs of certain lengths.

In this paper we show that for any network with total length $\mu $
and $l\geq 0$ leaf arcs, the value $V$ of the game (probability that the
attack is intercepted) is given by $V=\alpha /\left( \mu +l\alpha /2\right) $
when $\alpha $ is less than the minimum circuit length and also less than
twice the length of any leaf arc. So the game is completely solved on
\textit{any} network for sufficiently small positive $\alpha .$ If there are
no leaf arcs, the optimal patrol strategy reduces to a periodic cycle on the
network which covers every arc exactly twice. (We give a new proof that such
a cycle always exists.) Such a path is an efficient way of patrolling a
network.

Of course many networks, for example museum corridors, have
cul-de-sacs, which make them hard to patrol. Our general result, stated
above, solves this problem for short attack durations $\alpha ,$ but we also
have results for larger durations. For networks which have a tree structure,
we identify a useful technical property which implies that the value of the
game is given by $V=v^{\ast },$ where $v^{\ast }=\alpha /\left( \mu +\Lambda
\right) $ and $\Lambda $ is the total length of certain points near the leaf
nodes of the network. We conjecture that in fact $V=v^{\ast }$ for all
trees. We show that our technical property (and hence $V=v^{\ast })$ holds
for stars where no leaf arc has more than half the total length $\mu $ of
the star. Finally, we show that for stars with a single arbitrarily long arc
and the rest equal length short arcs, our conjecture $V=v^{\ast }$ holds.
Star networks are important and often occur at airports where there is a
central hub. The related problem of the ``uniformed patroller'' studied by
\cite{AlpernKatsikas19} and \cite{AlpernKatsikas21}, where the presence of the Patroller at the node chosen
for eventual attack can be detected by the Attacker, is studied in a spatial
context that can be viewed as a star network.

The knowledge of our results would be useful in designing networks
which are easier to patrol, as well as showing how to optimally patrol them.
Even when the network is given, one might add additional links between some
leaf nodes for the Patroller to use.  A useful extension to this problem
would be to make certain points of $Q$ more valuable than others, so that
successful attacks at such points are more costly to the Patroller and so
would need to be patrolled more intensively.

\subsection*{Acknowledgements} This material is based upon work supported by the National Science Foundation under Grant No. CMMI-1935826.

\section*{Appendix}

\paragraph{Proof of Lemma~\ref{lem:uniform}.}
The attack takes place during the time interval $J=\left[M,M+\alpha \right] $.
Since $S$ satisfies the unit speed condition (\ref{Lipshitz}), we have that $\lambda \left( S\left( J\right) \right) \leq \left\vert J\right\vert =\alpha $, where $\left\vert J\right\vert $ is the
length of $J$.
By the definition of the uniform attack strategy, the probability that the attack takes place in $S\left( J\right) $, and is thus intercepted, does not exceed $\lambda \left( S\left( J\right) \right) /\mu $, giving the claimed bound. 
\hfill$\Box$

\bigskip

\paragraph{Proof of Lemma~\ref{lemma:identifying}.}
First observe that the new metric $d^{\prime }$ will still have speed one.
If $S$ is a patrol on $Q$, then it satisfies (\ref{Lipshitz}) so%
\[
d^{\prime }\left( S\left( t\right) ,S\left( t^{\prime }\right) \right) \leq
d\left( S\left( t\right) ,S\left( t^{\prime }\right) \right) \leq \left\vert
t-t^{\prime }\right\vert ,\text{ }
\]%
which means that $S$ is still a patrol on $Q^{\prime },d^{\prime }$. On the
other hand, attacks on $Q^{\prime },d^{\prime }$ are the same as the attacks on $%
Q,d$. So the Patroller might have additional strategies whereas the Attacker
has no new strategies. Thus the new game can only be the same or better for
the Patroller, giving the main inequality. If the length of an arc is
decreased then the new metric satisfies the assumption  $0\leq d^{\prime
}\left( x,y\right) \leq d\left( x,y\right) $. Finally suppose $x$ and $y$
are identified, so that $Q^{\prime },d^{\prime }$ has the quotient topology.
For any points $z$ and $w$ in $Q-\left\{ x,y\right\} $ we have
\[
d^{\prime }\left( z,w\right) =\min \left\{ d\left( z,w\right) ,d\left(
z,x\right) +d\left( x,w\right) ,d\left( z,y\right) +d\left( y,w\right)\right\} \leq
d\left( z,w\right),
\]%
so the result follows from the first part of the proof. 
\hfill$\Box$

\bigskip

\paragraph{Proof of Lemma~\ref{lemma:paired}}

This proof mimics the usual proof of Euler's Theorem.
We first construct a circuit $C$ satisfying condition~(\ref{paired}),
which we call a $\ast$-circuit, using the following rules:

\begin{enumerate}
\item Start at any node $x$ and leave by any passage $P$ (we let $P'$ be the paired passage of $P$).

\item Always leave a node by an untraversed passage not paired with the arriving passage.

\item If, after arriving at a node, there are three untraversed passages with exactly
two of them paired, leave by one of this pair.

\item If, after arriving at node $x$, there are two untraversed passages, leave by passage $P^{\prime }$, if it is untraversed.

\item If there are no remaining untraversed passages after arriving at a node, stop.
\end{enumerate}

To simply obtain a circuit (not necessarily satisfying (2)) starting and ending at $x$, we would follow the usual method of simply leaving a node by any \emph{untraversed passage}, a simpler form of Rule~2. The existence of an untraversed passage (at any node other than the starting node $x$) follows from the fact the after arriving at a node an odd number of passages will have been traversed, so an odd number (hence not 0) are untraversed. We show that the full form of Rule~2 along with the other rules ensure that we can always leave a node in a way that satisfies (2) whether the node is the initial node $x$ or another node $y$.

We first check that after arriving at a node $y$ other than the starting node $x$, there cannot be only one remaining untraversed passage which is paired with the arriving passage. Since every node has even degree and there are no degree two nodes, the node $y$ must have been previously arrived at. After this previous arrival at $y$, there must have been three untraversed passages with exactly two of them paired. But Rule~3 ensures the circuit left by one of those two passages, so after arriving by the other one on the final visit, the last untraversed passage must have a different label.

To check that the final arriving passage at the initial node $x$ is not $P^{\prime }$, note that if $P^{\prime }$ had not been traversed before the penultimate visit to $x$, Rule~4 ensures that it will be traversed on that visit, and it will not be the final arriving passage.

If $C$ is a tour (contains all the arcs), we are done. Otherwise, since $Q$ is connected, there is a node $z$ with some passages in $C$ and some not in $C$ (see Figure~\ref{fig:Euler-proof-KP}).
Suppose that $C$ leaves $z$ beginning via passage $a$ and ends at $z$ via passage $b$. We create a new $\ast$-circuit starting at $z$, called $C^{\prime}$, using the same rules and using only passages not in $C$. Suppose $C^{\prime }$ begins with a passage called $d$ (which we can choose) and ends with a passage called $e$ (which we cannot control). The combined circuit $CC'$ which starts at $z$ and traverses $C$ and then $C'$ will satisfy (2) except possibly for the transitions $b,d$ and $e,a$ between the two circuits, so we need $d\neq b^{\prime }$ and $e\neq a^{\prime }$ (this means $d$ is not paired with $b$ and $e$ is not paired with $a)$. The arc $d$ is chosen as follows.

\begin{figure}[H]
\center
  \includegraphics[scale=0.3]{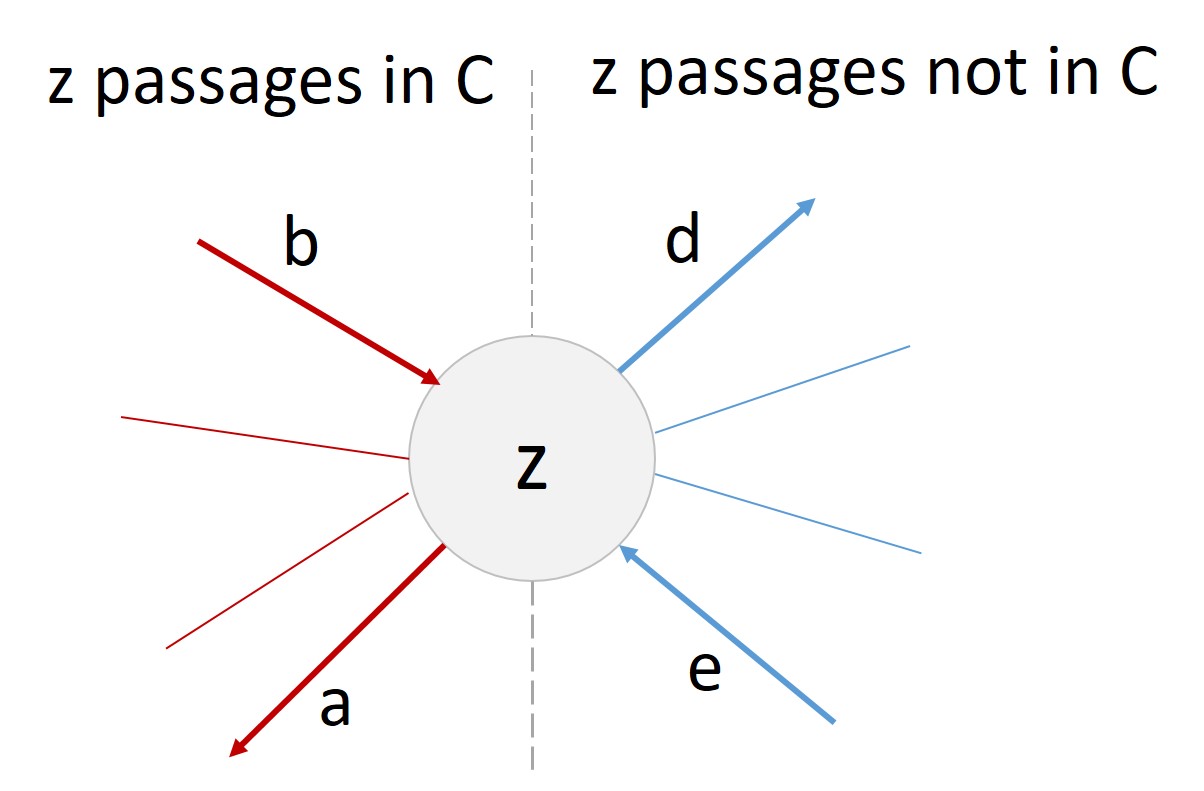}
  \caption{How to join two $\ast$-circuits at node $z$.}
  \label{fig:Euler-proof-KP}
\end{figure}

\begin{enumerate}
\item If $a^{\prime }$ is not in $C$, take $d=a^{\prime }$. This ensures
that $d=a^{\prime }\neq b^{\prime }$ since $a\neq b$. Also $e\neq
d=a^{\prime }$.

\item If $a^{\prime }$ is in $C$, take $d\neq b^{\prime }$. We know that
also $e\neq a^{\prime }$ because $a^{\prime }$ is in $C$.
\end{enumerate}

If the circuit $CC'$ is not a tour, we iteratively continue to add new circuits until we end up with a tour, noting that the process is guaranteed to end since every new circuit contains at least one new arc and there are a finite number of arcs. 
\hfill$\Box$

\bigskip

\paragraph{Proof of Lemma~\ref{lemma:spread_trees}.}
We fix a best response $S$ to the $E$-attack strategy, and show that the probability of interception is no more than $v^*$. To do this, we will define a new network $Q'$ of total length $\mu + \lambda(E)$ and a patrol $S'$ of $Q'$, and show that the probability $S$ intercepts the $E$-attack strategy on $Q$ is equal to the probability that $S'$ intercepts the uniform attack strategy (starting at time $t=M$) on $Q'$. The latter probability is at most $v^*$, by Lemma~\ref{lem:uniform}, so this will complete the proof.

The network $Q'$ is derived from $Q$ by replacing each component $E_i$ of $E$ with a loop $L_i$ of length $2 e_i$, where $e_i = \lambda(E_i)$. This is possible by the Leaf Condition, and clearly $\lambda(Q')=\mu + \lambda(E)$. Note that the probability the attack takes place on $E_i$ under the $E$-attack strategy is equal to the probability the attack takes place on $L_i$ under the uniform attack strategy. Since $S$ is a best response, we can assume that whenever it enters some component $E_i$, it proceeds directly to the leaf node of $E_i$, arriving at some time $t_1$, then leaves at a later time $t_2$, and returns directly to $E^c$. We will show later that we can assume $t_1=t_2$, so that $S$ performs tours of the components of $E$. We define the Patroller strategy $S'$ on $Q'$ by setting it equal to $S$ when $S$ is in $E^c$, and replacing any tour that $S$ performs of a component $E_i$ of $E$ in $Q$ with a tour of the loop $L_i$ in $Q'$.

Let $p_0$ be the probability the attack on $Q$ is intercepted by $S$, conditional on it taking place on $E^c$ and let $q_0$ be the corresponding conditional probability for $S'$ and $Q'$. Clearly, $p_0=q_0$. For every component $E_i$ of $E$, we also define $p_i$ to be the probability that the attack on $Q$ is intercepted by $S$, conditional on it taking place at the leaf node $x_i$ in the closure of $E_i$.  Similarly, we define $q_i$ to be the probability that the attack on $Q'$ is intercepted by $S'$ conditional on it taking place in $L_i$. It is sufficient to show that $p_i=q_i$ for each $i$.

The timing of the attacks on $Q$ is shown in Figure~\ref{fig:R-attack-v2-TL}. The first attack at $x_i$ finishes at time $M-e_i+\alpha$ and the last attack starts at time $M+e_i$. By the Leaf Condition, $M+e_i \le M-e_i+\alpha$, so $p_i=1$ if and only if the patrol visits $x_i$ in the time interval $[M+e_i, M-e_i+\alpha]$. Recall that $t_1$ and $t_2$ are the respective times that $S$ arrives at and leaves node $x_i$.

\begin{figure}[H]
    \centering
    \includegraphics[width = 4in]{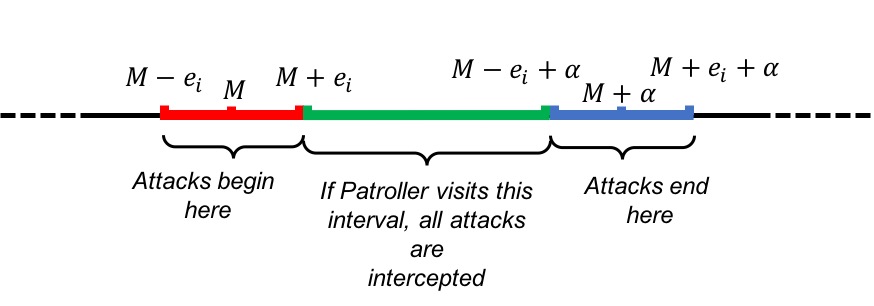}
    \caption{Timing of the attacks on $Q$.}
    \label{fig:R-attack-v2-TL}
\end{figure}

First suppose $p_i=1$. In this case, there is some time $t_0 \in \left[ M+e_i, M-e_i + \alpha \right]$ when the Patroller is at $x_1$, so we may as well assume that $t_1=t_2=t_0$ (otherwise we can replace $S$ with a patrol that dominates it). So that $S$ performs a tour of $E_i$ during the time interval $[t_0-e_i, t_0 + e_i] \subset [M, M+\alpha]$. This means that $S'$ also performs a tour of $L_i$ during this time interval, and therefore  $q_i=1$.

Now suppose that $p_i<1$. In this case, we must have either $t_1 > M - e_i + \alpha$ or $t_2 < M + e_i$. In the former case, the probability of an attack starting at $x_i$ after time $t_1$ is zero, so we can assume that $t_2=t_1$. In other words, $S$ performs a tour of $E_i$ in the time interval $[t_1-e_i,t_1+e_i]$, and $S'$ performs a tour of $L_i$ in the same time interval. If $t_1 \ge M + e_i + \alpha$ then $p_i=q_i=0$. If $t_1 \le M + e_i + \alpha$, then $S$ intercepts the attack if it starts at $x_i$ in the time interval $[t_1-\alpha, M+e_i]$, so $p_i=(M+e_i+\alpha - t_1)/(2e_i)$. Furthermore, $S'$ intercepts the attack if it takes place in $S([t_1-e_i, M+\alpha]) \subset L_i$, so $q_i=p_i$.

A similar argument holds for the case of $t_2 < M+e_i$ and this completes the proof. 
\hfill$\Box$


\bigskip
 \paragraph{Proof of Theorem~\ref{thm:x-n-star}.}
It is enough to show the statement is true for the case that $x>n$ and $2n\leq \alpha \leq 2x$. Assume the Attacker uses the symmetric-skewed attack strategy. Let $p_i$ be the probability that the attack is intercepted, conditional on it taking place at node $a_i$ for $i =1,..., n$. Let $q_j$ be the probability that the attack is intercepted, conditional on it taking place at $E^c$ at time $\alpha/2+2j$, for $j= 0,...,n-1$. Let $p_b$ be the probability that the attack is intercepted, conditional on it taking place at node $b$.

It is easy to see that if $p_b=1$, then the patrol must either stay at node $b$ until time $\alpha+2(n-1)$ or arrive at node $b$ at time $\alpha$ or earlier (then stay there). In both cases, we have $q_j=0$ for all $j$ since the patrol cannot be in $\lambda(E^c)$ during the time interval $[\alpha/2, 3\alpha/2+2(n-1)]$. Similarly, we have $p_i=0$ for all $i$. Thus, the probability of interception is $\alpha/(\mu+\lambda(E))$.

Next, suppose $p_b <1$ and the patrol stays at node $b$ until some time $t< \alpha+2(n-1)$. We will split this case into two subcases: $2n\leq \alpha \leq \mu$ and $\mu \leq \alpha \leq 2x$.

Considering the first subcase, $2n\leq \alpha \leq \mu$, we assume the patrol arrives node $v$ (the right side of $E^c$) at time $r=t+\alpha/2$, and $q_0$ will be bounded by the function $\gamma(r)$ as below.

If $\lambda(E^c) \geq \alpha$,
\begin{equation*}
    \gamma(r)=
    \begin{cases}
    \frac{r+\lambda(E^c)-\alpha/2}{\lambda(E^c)} & \text{if}\ \ 0\leq r \leq 3\alpha/2-\lambda(E^c),\\
    \frac{\alpha}{\lambda(E^c)} & \text{if}\ \ 3\alpha/2-\lambda(E^c) \leq r \leq \alpha/2,\\
    \frac{3\alpha/2-r}{\lambda(E^c)}& \text{if}\ \ \alpha/2 \leq r \leq 3\alpha/2,\\
    0 & \text{if}\ \ 3\alpha/2 \leq r.
    \end{cases}
\end{equation*}

If $\lambda(E^c) \leq \alpha$,
\begin{equation*}
    \gamma(r)=
    \begin{cases}
    \frac{r+\lambda(E^c)-\alpha/2}{\lambda(E^c)} & \text{if}\ \ \alpha/2 - \lambda(E^c)\leq r \leq \alpha/2,\\
    1 & \text{if}\ \ \alpha/2 \leq r \leq 3\alpha/2- \lambda(E^c),\\
    \frac{3\alpha/2-r}{\lambda(E^c)}& \text{if}\ \ 3\alpha/2- \lambda(E^c) \leq r \leq 3\alpha/2,\\
    0 & \text{if}\ \ 3\alpha/2 \leq r .
    \end{cases}
\end{equation*}
Thus, for $j=1,...,n-1$, the probability $q_j$ is bounded by $\gamma(r-2j)$.

Without loss of generality, we assume the patrol arrives at node $a_1$ at time $s=t+x+1$ then moves within all $a_i$ ($i=1,...,n$) thereafter. If the patrol visits every other node $a_i$ ($i \neq 1$) before  returns to $a_1$, it takes time $2n\leq \alpha$. So, all attacks at node $a_1$ happening from time $s-\alpha$ will be intercepted and $p_1$ is bounded above by
\begin{equation*}
    \delta(s)=
    \begin{cases}
    1 & \text{if}\ \ 0\leq s \leq \alpha+ n-1,\\
    \frac{(2\alpha+n-1)-s}{\alpha} & \text{if}\ \ \alpha+n-1 \leq s \leq 2\alpha+n-1.
    \end{cases}
\end{equation*}
If the patrol moves directly from $a_1$ to $a_2$, then $p_2$ is bounded above by $\delta(s+2)$. Upper bounds for the other $p_i$ can be calculated in the same way.

A patrol that leaves node $b$ at time $t$, arrives at $E^c$ at time $t+\alpha/2$, crosses $E^c$ to reach node $a_1$ at time $t+x+1$, then moves within $a_i$ thereafter has interception probability $p(t)$ given by
\begin{align*}
p(t) =&\frac{\alpha}{\mu +\lambda(E)} f(t)+ \frac{\lambda(E^c)}{n(\mu+\lambda(E))}\sum_{j=0}^{n-1}\gamma(t+\frac{\alpha}{2}-2j)\\
&+\frac{\alpha}{n(\mu +\lambda(E))}\sum_{i=1}^{n}\delta(t+x+1+2(i-1))\\
&= \frac{\alpha}{\mu+\lambda(E)}.
\end{align*}

We now consider the second subcase, $\mu \leq \alpha \leq 2x$. In this case, $\lambda(E)=\mu$ so that $E^c$ is empty and there are no middle attacks. We assume the patrol leaves node $b$ at time $t$, visits all nodes $a_i$ at time $t+x+1+2(i-1)$ then returns to $x$ at time $t+2\mu$.

Let $s$ be the time the patrol arrives at node $a_i$. Then $p_i$ is bounded by the function $g(s)$ given below.
\begin{equation*}
    g(s)=
    \begin{cases}
    \frac{s-(\alpha-x-1)}{2(x+n)-\alpha} & \text{if}\ \ \alpha-x-1\leq s \leq x+2(n-1)+1,\\
    1 & \text{if}\ \ x+2(n-1)+1\leq s \leq 2\alpha-x-1,\\
    \frac{\alpha+x+2(n-1)+1-s}{2(x+n)-\alpha}& \text{if}\ \ 2\alpha-x-1 \leq s \leq \alpha+x+2(n-1)+1.
    \end{cases}
\end{equation*}

Since the patrol stays leaves node $b$ at time $t$ and returns at $t+2\mu$, the upper bound $h(t)$ of $p_b$ is
$$h(t)= f(t)+1-f(t+2\mu-\alpha).$$

So, the interception probability $p(t)$ of the patrol is
$$p(t)=\frac{\alpha}{2\mu}h(t)+ \frac{2\mu-\alpha}{n\mu}\sum_{i=1}^n g(t+x+1+2(i-1))=\frac{\alpha}{2\mu} = \frac{\alpha}{\mu+\lambda(E)}.$$

\hfill$\Box$

\end{document}